\documentclass[a4]{article}
\usepackage{bm}
\usepackage[all]{xy}
\usepackage{graphicx}
\usepackage{verbatim}
\usepackage{wrapfig}
\usepackage{ascmac}
\usepackage{makeidx}
\usepackage{amscd}
\usepackage{url}
\usepackage{comment}
\usepackage{enumerate}
\usepackage{here}
\usepackage{latexsym}
\usepackage{array}
\usepackage{booktabs}
\usepackage{multirow}
\usepackage{mathrsfs}
\usepackage{geometry}
\usepackage{fancyhdr}
\renewcommand{\title}[1]{

\begin{center} \Large \bf #1 \end{center}
}

\renewcommand{\author}[2]{
 \begin{center} #1  \vspace{3mm} \\
  #2 \\
 \end{center}
\addvspace{\baselineskip}
}

\usepackage{amssymb}
\usepackage{amsmath}

\usepackage{amsthm}
\newtheorem{theorem}{Theorem}[section]
\newtheorem{proposition}[theorem]{Proposition}

\newtheorem{lemma}[theorem]{Lemma}

\theoremstyle{definition}

\theoremstyle{remark}

\newtheorem*{fact}{Fact}

\makeatletter
\@addtoreset{equation}{section}

\makeatother

\begin{document}
\baselineskip 5mm
\title{Exact solution of the $\Phi_{2}^{3}$ finite matrix model}
\author{${}^1$ Naoyuki Kanomata and~ ${}^1$ Akifumi Sako}
{
${}^1$  Tokyo University of Science,\\ 1-3 Kagurazaka, Shinjuku-ku, Tokyo, 162-8601, Japan
}
\noindent

\vspace{1cm}

\abstract{\vspace{1mm}

We find the exact solutions of the $\Phi_{2}^{3}$ finite matrix model (Grosse-Wulkenhaar model). 
In the $\Phi_{2}^{3}$ finite matrix model, multipoint correlation functions are expressed as $G_{|a_{1}^{1}\ldots a_{N_{1}}^{1}|\ldots|a_{1}^{B}\ldots a_{N_{B}}^{B}|}$.
The $\displaystyle \sum_{i=1}^{B}N_{i}$-point function denoted by $G_{|a_{1}^{1}\ldots a_{N_{1}}^{1}|\ldots|a_{1}^{B}\ldots a_{N_{B}}^{B}|}$ is given by the sum over all Feynman diagrams (ribbon graphs) on Riemann surfaces with $B$-boundaries, and each $|a^{i}_{1}\cdots a^{i}_{N_{i}}|$ corresponds to the Feynman diagrams having $N_{i}$-external lines from the $i$-th boundary. 
It is known that any $G_{|a_{1}^{1}\ldots a_{N_{1}}^{1}|\ldots|a_{1}^{B}\ldots a_{N_{B}}^{B}|}$ can be expressed using $G_{|a^{1}|\ldots|a^{n}|}$ type $n$-point functions. 
Thus we focus on rigorous calculations of $G_{|a^{1}|\ldots|a^{n}|}$. 
The formula for $G_{|a^{1}|\ldots|a^{n}|}$ is obtained, and it is achieved by using the partition function $\mathcal{Z}[J]$ calculated by the Harish-Chandra-Itzykson-Zuber integral. 
We give $G_{|a|}$, $G_{|ab|}$, $G_{|a|b|}$, and $G_{|a|b|c|}$ as the specific simple examples. All of them are described by using the Airy functions.
}

\section{Introduction}

Matrix models were well studied in the 1980s and 1990s in the context of non-critical string theories\cite{DiFrancesco:1993cyw}. 
Perturbative expansions of the matrix models can be interpreted by using random simplicial decompositions of two-dimensional surfaces.
Each Feynman diagram in perturbative expansions of path integrals represents a corresponding simplicial decomposition of a two-dimensional surface. 
In particular, Feynman diagrams of $\Phi^{3}$ matrix models can be regarded as triangulations of two-dimensional surfaces.
The sum over two-dimensional surfaces corresponds to path integrals of two-dimensional quantum gravity theories\cite{DiFrancesco:1993cyw},\cite{Gross:1989aw},\cite{1991}. 
$1/N$ expansions of the matrix models of size $N$ are equivalent to the genus expansion of the two-dimensional surface\cite{Gross:1989aw}. 
The double scaling limit is a limit that incorporates contributions from any genus $g$ surfaces\cite{DiFrancesco:1993cyw},\cite{Gross:1989aw},\cite{Brezin:1990rb},\cite{Douglas:1989ve},\cite{1991}. 
Calculating the genus expansion of the two-dimensional surfaces in this limit gives a fully non-perturbative solution for the two-dimensional gravity theory\cite{1991}.\bigskip

Fukuma, Kawai, and Nakayama proved that the Virasoro constraint condition is equivalent to the condition that the solution of the KdV hierarchy satisfies the string equation\cite{Fukuma:1990jw},\cite{Zhou:2013kka}. 
Dijkgraaf, Verlinde, and Verlinde also derived the similar results independently\cite{Dijkgraaf:1990rs}.
Witten showed that the Witten-Kontsevich $\tau$-function satisfies the string equation\cite{Zhou:2013kka},\cite{Witten:1990hr}. 
Furthermore, Witten conjectured that the Witten-Kontsevich $\tau$-function is the $\tau$ function of the KdV hierarchy\cite{Zhou:2013kka},\cite{Witten:1990hr}. 
Using the $\Phi^{3}$ matrix model, the proof of this conjecture was done by Kontsevich\cite{Kontsevich:1992ti},\cite{Zhou:2013kka}.\bigskip

We focus on the Grosse-Wulkenhaar $\Phi^{3}$ model\cite{Grosse:2016pob},\cite{Grosse:2016qmk},\cite{Hock:2020rje} (the Kontsevich's $\Phi^{3}$ matrix model).
We refer to the Grosse-Wulkenhaar $\Phi^{3}$ model simply as the $\Phi^{3}$ model in the following.
Quantum field theories on noncommutative spaces such as Moyal spaces have given a new perspective to matrix models.
The problem of UV/IR-mixing generally arises when considering the quantum field theories on noncommutative spaces.
Grosse and Wulkenhaar modified the $\Phi^{4}$ theory on the four-dimensional Moyal space by adding harmonic oscillator potentials to the action\cite{Grosse:2004yu}.
Grosse and Wulkenhaar also proved that this model is renormalizable at all orders of perturbation theories\cite{Grosse:2012uv}.
The $\Phi^{3}$ matrix model on Moyal spaces was studied by Grosse and Steinacker\cite{Grosse:2005ig},\cite{Grosse:2006qv},\cite{Grosse:2006tc}.
In particular, the $n$-point functions of the $\Phi^{3}$ matrix model in the large $N, V$ limit were calculated in the previous studies\cite{Grosse:2016pob},\cite{Grosse:2016qmk} by Grosse, Wulkenhaar, and one of the authors of this paper.
Any $n$-point function of the $\Phi^{3}$ matrix model was calculated by solving the Schwinger-Dyson equation exactly by using the Ward-Takahashi identity.
The Swiss cheese limit picks up only the Genus $0$ contribution while preserving the boundary components.
The Schwinger-Dyson equation obtained in this limit is the integral equation corresponding to the Riemann-Hilbert problem.
The Schwinger-Dyson equation for the $1$ point function in the Swiss cheese limit coincides with the one in Makkeenko-Semenoff\cite{Makeenko:1991ec}.
Using the hierarchy of the Schwinger-Dyson equation and this exact solution, any $n$-point function in the large $N, V$ limit was obtained\cite{Grosse:2016pob},\cite{Grosse:2016qmk}.
Afterward, the planar $2$-point function of the Grosse-Wulkenhaar type $\Phi^{4}$ model in large $N,V$ limit was solved exactly by Grosse, Wulkenhaar, and Hock\cite{Grosse:2019jnv}, $n$-point functions were solved by Wulkenhaar, and Hock\cite{Hock:2021tbl}.
Wulkenhaar, Branahl, and Hock found blobbed topological recursion of the $\Phi^{4}$ model\cite{Branahl:2020yru},\cite{Borot:2015hna}. \bigskip

In this paper, we find exact solutions of the $\Phi^{3}$ finite matrix model. 
In calculating the partition function $\mathcal{Z}[J]$, the integration of the off-diagonal elements of the Hermitian matrix is calculated using the Harish-Chandra-Itzykson-Zuber integral\cite{Itzykson:1979fi},\cite{T.Tao},\cite{Harish-Chandra:1957dhy},\cite{Zinn-Justin:2002rai}.
On the other hand, the integral of the diagonal elements of the Hermitian matrix is done by using the Airy functions.
We use the result to calculate $G_{|a^{1}|\ldots|a^{n}|}$ type $n$-point functions. 
The exact solutions of $G_{|a^{1}|\ldots|a^{n}|}$ type $n$-point functions can be obtained by calculating the n-th derivative $\partial^{n}/\partial J_{a^{1}a^{1}}\cdots\partial J_{a^{n}a^{n}}$ of  $\log\mathcal{Z}[J]$ with the external field $J$ as a diagonal matrix.
We apply the formula from the previous study \cite{Hardy} to this formula.
As a result, we succeed to find the exact solutions for $G_{|a^{1}|\ldots|a^{n}|}$ type $n$-point functions.
Since $n$-point functions $G_{|a_{1}^{1}\ldots a_{N_{1}}^{1}|\ldots|a_{1}^{B}\ldots a_{N_{B}}^{B}|}$ can be expressed by using  $G_{|a^{1}|\ldots|a^{n}|}$ type $n$-point functions, we can obtain all the exact solutions of the $\Phi^{3}$ finite matrix model. 
We also give $G_{|a|}$, $G_{|ab|}$, $G_{|a|b|}$, and $G_{|a|b|c|}$ as concrete functions as examples.\bigskip

This paper is organized as follows.
In Section \ref{sec2}, we review the $\Phi^3$ matrix model.
In Section \ref{sec3}, we carry out the path integral of the partition function $\mathcal{Z}[J]$. 
In Section \ref{sec4}, using the result in Section \ref{sec3}, we calculate the exact solutions of $G_{|a|}$, and $G_{|ab|}$.
In Section \ref{sec5}, we derive the exact solutions of the $n$-point functions. 
In Section \ref{sec6}, $G_{|a|b|}$ is given as the simple examples.
In Section \ref{sec7}, we summarize this paper.


\section{Setup of $\Phi_{2}^{3}$ Matrix Model}\label{sec2}
In this section, we review the $\Phi_{2}^{3}$ matrix model based on the previous studies\cite{Grosse:2016pob},\cite{Grosse:2016qmk},\cite{Hock:2020rje}, and we determine the notation in this paper.

Let $\Phi=(\Phi_{ij})$ be a Hermitian matrix for $i,j=1,2,\ldots,N$ and $E_{m-1}$ be a discretization of a monotonously increasing differentiable function $e$ with $e(0)=0$, 
\begin{align}
E_{m-1}=&\mu^{2}\left(\frac{1}{2}+e\left(\frac{m-1}{\mu^{2}V}\right)\right),
\end{align}
where $\mu^{2}$ is a squared mass and $V$ is a real constant. Let $E=(E_{m-1}\delta_{mn})$ be a diagonal matrix for $m,n=1,\ldots,N$. Let us consider the following action:
\begin{align}
S[\Phi]=iV\mathrm{tr}\left(E\Phi^{2}+\kappa\Phi+\frac{\lambda}{3}\Phi^{3}\right),
\end{align}
where  $\kappa$ is a renormalization constant (real), $\lambda$ is a coupling constant that is non-zero real, and $i=\sqrt{-1}$. Compared to the paper\cite{Grosse:2016pob}, the difference is that $V$ is replaced with $iV$. By the diagonal matrix $E$ that is not proportional to the unit matrix in general, there is no symmetry for the unitary transformation in $\Phi\rightarrow U\Phi U^{*}$. Here $U$ is a unitary matrix, and $U^{*}$ is its Hermitian conjugate. 

Let $J=(J_{mn})$ be a Hermitian matrix for $m,n=1,\ldots,N$ as an external field. Let $\mathcal{D}\Phi$ be the integral measure, 
\begin{align}
\displaystyle\mathcal{D}\Phi:=&\prod_{i=1}^{N}d\Phi_{ii}\prod_{1\leq i<j\leq N}d\mathrm{Re}\Phi_{ij}d\mathrm{Im}\Phi_{ij},
\end{align}   
where each variable is divided into real and imaginary parts $\Phi_{ij}=\mathrm{Re}\Phi_{ij}+i\mathrm{Im}\Phi_{ij}$ with $\mathrm{Re}\Phi_{ij}=\mathrm{Re}\Phi_{ji}$ and $\mathrm{Im}\Phi_{ij}=-\mathrm{Im}\Phi_{ji}$. Let us consider the following partition function:
\begin{align}
\mathcal{Z}[J]:=&\int\mathcal{D}\Phi\exp\left(-S[\Phi]+iV\mathrm{tr}(J\Phi)\right)\notag\\
=&\int \mathcal{D}\Phi \exp\left(-iV\mathrm{tr}\left(E\Phi^2+\kappa\Phi+\frac{\lambda}{3}\Phi^3\right)\right)\exp\left(iV\mathrm{tr}\left(J\Phi\right)\right).\label{Z[J]}
\end{align}

Using $\displaystyle\log\frac{ \mathcal{Z}[J]}{\mathcal{Z}[0]}$, the $\displaystyle \sum_{i=1}^{B}N_{i}$-point function $G_{|a_{1}^{1}\ldots a_{N_{1}}^{1}|\ldots|a_{1}^{B}\ldots a_{N_{B}}^{B}|}$ is defiend as
\begin{align}
\log\frac{ \mathcal{Z}[J]}{\mathcal{Z}[0]}
:=\sum_{B=1}^\infty \sum_{1\leq N_1 \leq \dots \leq
  N_B}^\infty
\sum_{p_1^1,\dots,p^B_{N_B} =0}^{\mathcal{N}} \!\!\!\!
(iV)^{2-B}
&\frac{G_{|p_1^1\dots p_{N_1}^1|\dots|p_1^B\dots p^B_{N_B}|}
}{S_{(N_1,\dots ,N_B)}}
\prod_{\beta=1}^B \frac{\mathbb{J}_{p_1^\beta\dots
    p^\beta_{N_\beta}}}{N_\beta},
\label{logZ}
\end{align}
where $N_{i}$ is the identical valence number for $i=1,\ldots,B$, $\displaystyle\mathbb{J}_{p_1\dots p_{N_{i}}}:=\prod_{j=1}^{N_{i}} J_{p_jp_{j+1}}$ with $N_{i}+1\equiv 1$, $(N_1,\dots,N_B)=(\underbrace{N'_1,\dots,N'_1}_{\nu_1},\dots,
\underbrace{N'_s,\dots,N'_s}_{\nu_s})$, and $\displaystyle S_{(N_1,\dots ,N_B)}=\prod_{\beta=1}^{s} \nu_{\beta}!$.


\section{Calculation of Partition Function $\mathcal{Z}[J]$}\label{sec3}
In this section, we perform the integration of the partition function $\mathcal{Z}[J]$ by dividing the Hermitian matrix into its diagonal and off-diagonal elements. The off-diagonal elements of the Hermitian matrix in the partition function $\mathcal{Z}[J]$ are integrated using the Harish-Chandra-Izykson-Zuber integral \cite{Itzykson:1979fi},\cite{T.Tao},\cite{Harish-Chandra:1957dhy},\cite{Zinn-Justin:2002rai}, and the integration of the diagonal elements of the Hermitian matrix in the partition function $\mathcal{Z}[J]$ is performed by using Airy functions. The calculations are essentially in the line with the calculations of Kontsevich \cite{Kontsevich:1992ti}. We write without omitting details because the results are different due to the presence of external fields $J$ and the renormalization term $\kappa$.\bigskip

We introduce a Hermitian matrix $\displaystyle\widetilde{E}=(\widetilde{E_{m}}\delta_{mn})=\frac{1}{\lambda}E=\left(\frac{E_{m-1}}{\lambda}\delta_{mn}\right)$ for $m,n=1,\ldots,N$ and $\displaystyle\frac{\kappa}{\lambda}=\widetilde{\kappa}$. Note that the indices are shifted i.e. $\widetilde{E}=diag\left(\widetilde{E}_{1},\cdots,\widetilde{E}_{N}\right)$ and $E=diag\left(E_{0},\cdots,E_{N-1}\right)$. Then $\mathcal{Z}[J]$ is written as
\begin{align}
\mathcal{Z}[J]=&\int \mathcal{D}\Phi \exp\left(-2i\lambda V\mathrm{tr}\left(\frac{\widetilde{E}\Phi^2}{2}+\frac{\widetilde{\kappa}\Phi}{2}+\frac{1}{3!}\Phi^3\right)\right)\exp\left(iV\mathrm{tr}\left(J\Phi\right)\right).
\end{align}
We introduce a new variable $X$ by $\Phi=X-\widetilde{E}$. Here $X=(X_{mn})$ is a Hermitian matrix, too. We do a change of variables of the integral measure $\mathcal{D}\Phi$ as 
\begin{align}
d\Phi_{ij}=&\sum_{m,n=1}^{N}\frac{\partial\Phi_{ij}}{\partial X_{mn}}dX_{mn}=dX_{ij}.
\end{align}
By the variable transformation $\displaystyle\mathrm{tr}\left(\frac{\widetilde{E}\Phi^2}{2}+\frac{\widetilde{\kappa}\Phi}{2}+\frac{1}{3!}\Phi^3\right)$ is
\begin{align}
\mathrm{tr}\left(\frac{\widetilde{E}\Phi^2}{2}+\frac{\widetilde{\kappa}\Phi}{2}+\frac{1}{3!}\Phi^3\right)=&\mathrm{tr}\left(\frac{(X)^3-3(\widetilde{E})^2X+2(\widetilde{E})^3+3\widetilde{\kappa}X-3\widetilde{\kappa}\widetilde{E}}{6}\right),
\end{align}
then $\mathcal{Z}[J]$ is given as
\begin{align}
\mathcal{Z}[J]=&\int \mathcal{D}X\exp\left(-2i\lambda V\mathrm{tr}\left(\frac{(X)^3-3(\widetilde{E})^2X+2(\widetilde{E})^3+3\widetilde{\kappa}X-3\widetilde{\kappa}\widetilde{E}}{6}\right)\right)\exp\left(iV\mathrm{tr}\left(JX-J\widetilde{E}\right)\right)\notag\\
=&\exp\left(-i\lambda Vtr\left(\frac{2}{3}(\widetilde{E})^{3}-\widetilde{\kappa}\widetilde{E}+\frac{1}{\lambda}J\widetilde{E}\right)\right)\int\mathcal{D}X\exp\left(+i\lambda Vtr\left(-\widetilde{\kappa}X+\frac{1}{\lambda}JX-\frac{1}{3}X^{3}+(\widetilde{E})^{2}X\right)\right)\notag\\
=&\exp\left(-i\lambda Vtr\left(\frac{2}{3}(\widetilde{E})^{3}-\widetilde{\kappa}\widetilde{E}+\frac{1}{\lambda}J\widetilde{E}\right)\right)\int\mathcal{D}X\exp\left(-i\frac{\lambda  V}{3}\mathrm{tr}(X^3)\right)\exp\left(i\lambda V\widetilde{\kappa}\mathrm{tr}\{(M-I+K)X\}\right).\label{XY}
\end{align}
Here  $M=\displaystyle\frac{(\widetilde{E})^2}{\widetilde{\kappa}}=\frac{E^{2}}{\lambda\kappa}$, $K=\displaystyle\frac{J}{\kappa}$, and $I$ is the unit matrix. Note that 
\begin{align*}
\mathcal{D}X=&\left(\prod_{i=1}^{N}dx_{i}\right)\left(\prod_{1\leq k<l\leq N}(x_{l}-x_{k})^{2}\right)dU,
\end{align*}
where $x_{i}$ is the eigenvalues of $X$ for $i=1,\cdots,N$, $dU$ is the Haar probability measure of the unitary group $U(N)$, and $U$ is the unitary matrix which diagonalize $X$\cite{M. L. Mehta}\cite{Eynard:2015aea}. Then (\ref{XY}) can be rewritten as the following:
\begin{align}
\mathcal{Z}[J]=&\exp\left(-i\lambda Vtr\left(\frac{2}{3}(\widetilde{E})^{3}-\widetilde{\kappa}\widetilde{E}+\frac{1}{\lambda}J\widetilde{E}\right)\right)\notag\\
&\int\left(\prod_{i=1}^{\mathrm{N}}dx_{i}\exp\left(-i\frac{\lambda  V}{3}x^{3}_{i}\right)\right)\left(\prod_{1\leq k<l\leq N}(x_{l}-x_{k})^{2}\right)\int dU\exp\left(i\lambda V\widetilde{\kappa}\mathrm{tr}\{(M-I+K)U\widetilde{X}U^{*}\}\right),\label{YX}
\end{align}
where $\widetilde{X}$ is the diagonal matrix $\widetilde{X}=U^{*}XU$. We use the following formula.

\begin{fact}

The Harish-Chandra-Izykson-Zuber integral \cite{Itzykson:1979fi},\cite{T.Tao},\cite{Harish-Chandra:1957dhy},\cite{Zinn-Justin:2002rai} for the unitary group $U(n)$ is
\begin{align}
\int_{U(n)}\exp\left(t\mathrm{tr}\left(AUBU^{*}\right)\right)dU=&c_{n}\frac{\displaystyle\det_{1\leq i,j\leq n}\left(\exp\left(t\lambda_{i}(A)\lambda_{j}(B)\right)\right)}{t^{\frac{(n^{2}-n)}{2}}\displaystyle\Delta(\lambda(A))\displaystyle\Delta(\lambda(B))}.\label{Itzykson}
\end{align}
Here $A=(A_{ij})$, and $B=(B_{ij})$ are some Hermitian matrices whose eigenvalues denoted by $\lambda_{i}(A)$ and $\lambda_{i}(B)$ $(i=1,\cdots,n)$, respectively. $t$ is the non-zero complex parameter, $\displaystyle\Delta(\lambda(A)):=\prod_{1\leq i<j\leq n}(\lambda_{j}(A)-\lambda_{i}(A))$ is the Vandermonde determinant, and $\displaystyle c_{n}:=\left(\prod_{i=1}^{n-1}i!\right)\times\pi^{\frac{n(n-1)}{2}}$ is the constant. $\left(\exp\left(t\lambda_{i}(A)\lambda_{j}(B)\right)\right)$ is the $n\times n$ matrix with the $i$-th row and the $j$-th column being $\exp\left(t\lambda_{i}(A)\lambda_{j}(B)\right)$.

\end{fact}

Applying the Harish-Chandra-Izykson-Zuber integral (\ref{Itzykson}) to $\displaystyle\int dU\exp\left(i\lambda V\widetilde{\kappa}\mathrm{tr}\{(M-I+K)U\widetilde{X}U^{*}\}\right)$ in (\ref{YX}), the result is
\begin{align}
\displaystyle\int dU\exp\left(i\lambda V\widetilde{\kappa}\mathrm{tr}\{(\mathrm{M}-\mathrm{I}+\mathrm{K})U\widetilde{X}U^{*}\}\right)=&\frac{C}{N!}\frac{\displaystyle\det_{1\leq i,j\leq N}\exp\left(i\lambda V\widetilde{\kappa} x_{i}s_{j}\right)}{\displaystyle\prod_{i<j}(x_{j}-x_{i})\prod_{i<j}(s_{j}-s_{i})},
\end{align}
where $s_{t}$ is the eigenvalues of the matrix $M-I+K$ for $t=1,\cdots,N$ and $\displaystyle C=\left(\prod_{p=1}^{N}p!\right)\times\left(\frac{\pi}{i\lambda V\widetilde{\kappa}}\right)^{\frac{N(N-1)}{2}}$. $\left(\exp\left(i\lambda V\widetilde{\kappa} x_{i}s_{j}\right)\right)$ denotes the $N\times N$ matrix with the $i$-th row and the $j$-th column being $\exp\left(i\lambda V\widetilde{\kappa} x_{i}s_{j}\right)$. Then the partition function $\mathcal{Z}[J]$ is described as
\begin{align}
\mathcal{Z}[J]=&\frac{C}{N!}\exp\left(-i\lambda Vtr\left(\frac{2}{3}(\widetilde{E})^{3}-\widetilde{\kappa}\widetilde{E}+\frac{1}{\lambda}J\widetilde{E}\right)\right)\frac{1}{\displaystyle\prod_{1\leq t<u\leq N}(s_{u}-s_{t})}\notag\\
&\int\left(\prod_{i=1}^{N}dx_{i}\exp\left(-i\frac{\lambda V}{3}x^{3}_{i}\right)\right)\left(\prod_{1\leq k<l\leq N}(x_{l}-x_{k})\right)\displaystyle\det_{1\leq m,n\leq N}\exp\left(i\lambda V\widetilde{\kappa} x_{m}s_{n}\right)\notag\\
=&C\exp\left(-i\lambda Vtr\left(\frac{2}{3}(\widetilde{E})^{3}-\widetilde{\kappa}\widetilde{E}+\frac{1}{\lambda}J\widetilde{E}\right)\right)\frac{1}{\displaystyle\prod_{1\leq t<u\leq N}(s_{u}-s_{t})}\notag\\
&\int\left(\prod_{i=1}^{N}dx_{i}\exp\left(-i\frac{\lambda  V}{3}x^{3}_{i}\right)\displaystyle\exp\left(i\lambda V\widetilde{\kappa} x_{i}s_{i}\right)\right)\prod_{1\leq k<l\leq N}(x_{l}-x_{k}).\label{aaaa}
\end{align}
Here we use the following result at the second equality in (\ref{aaaa}):
\begin{align}
&\int\left(\prod_{i=1}^{N}dx_{i}\exp\left(-i\frac{\lambda V}{3}x^{3}_{i}\right)\right)\left(\prod_{1\leq k<l\leq N}(x_{l}-x_{k})\right)\displaystyle\det_{1\leq m,n\leq N}\exp\left(i\lambda V\widetilde{\kappa} x_{m}s_{n}\right)\notag\\
=&\sum_{\sigma\in S_{N}}\int\left(\prod_{i=1}^{N}dx_{i}\exp\left(-i\frac{\lambda V}{3}x^{3}_{i}\right)\right)\left(\prod_{1\leq k<l\leq N}(x_{l}-x_{k})\right)(-1)^{\sigma}\left(\prod_{j=1}^{N}e^{i\lambda V\widetilde{\kappa}x_{\sigma(j)}s_{j}}\right)\notag\\
=&\sum_{\sigma\in S_{N}}\int\left(\prod_{i=1}^{N}dx_{i}\exp\left(-i\frac{\lambda V}{3}x^{3}_{i}\right)\right)\left(\prod_{1\leq k<l\leq N}(x_{l}-x_{k})\right)(-1)^{\sigma}(-1)^{\sigma}\left(\prod_{j=1}^{N}e^{i\lambda V\widetilde{\kappa}x_{j}s_{j}}\right)\notag\\
=&N!\int\left(\prod_{i=1}^{N}dx_{i}\exp\left(-i\frac{\lambda  V}{3}x^{3}_{i}\right)\displaystyle\exp\left(i\lambda V\widetilde{\kappa} x_{i}s_{i}\right)\right)\prod_{1\leq k<l\leq N}(x_{l}-x_{k}).
\end{align}

In the transformation of the above equation from the second line to the third line, we changed variables as $x_{\sigma(i)}\mapsto x_{i}~(i=1,\cdots,N)$.

Using $\displaystyle\prod_{1\leq k<l\leq N}(x_{l}-x_{k})=\det_{1\leq k,l\leq N}\left(x_{k}^{l-1}\right)$, we calculate the remaining integral in the right hand side in (\ref{aaaa}) as 
\begin{align}
&\int_{-\infty}^{\infty}\left(\prod_{i=1}^{N}dx_{i}\exp\left(-i\frac{\lambda V}{3}x_{i}^{3}\right)\exp\left(i\lambda V\widetilde{\kappa}s_{i}x_{i}\right)\right)\det_{1\leq k,l\leq\mathrm{N}}\left(x^{k-1}_{l}\right)\notag\\
=&\sum_{\sigma\in S_{\mathrm{N}}}\mathrm{sgn}\sigma\prod_{i=1}^{N}\phi_{\sigma(i)}(s_{i})\notag\\
=&\det_{1\leq i,j\leq N}\left(\phi_{i}(s_{j})\right),\label{aaa}
\end{align}
where $\displaystyle\phi_{k}(z)$ is defined by
\begin{align}
\displaystyle\phi_{k}(z)=&\int_{-\infty}^{\infty}dx\hspace{2mm}x^{k-1}\exp\left(-i\frac{\lambda V}{3}x^3+iV\kappa xz\right),
\end{align}
and $(\phi_{i}(s_{j}))$ is the $N\times N$ matrix with the $i$-th row and the $j$-th column being $\phi_{i}(s_{j})$. Summarizing the results (\ref{aaaa}) and (\ref{aaa}), we obtain the following:

\begin{proposition}\label{Pro3.1}

Let $\mathcal{Z}[J]$ be the partition function of the $\Phi_{2}^{3}$ matrix model given by (\ref{Z[J]}). Then, $\mathcal{Z}[J]$ is given as
\begin{align*}
\mathcal{Z}[J]=&C\exp\left(-i\lambda Vtr\left(\frac{2}{3}(\widetilde{E})^{3}-\widetilde{\kappa}\widetilde{E}+\frac{1}{\lambda}J\widetilde{E}\right)\right)\frac{\displaystyle\det_{1\leq i,j\leq N}\left(\phi_{i}(s_{j})\right)}{\displaystyle\prod_{1\leq t<u\leq\mathrm{N}}(s_{u}-s_{t})}.
\end{align*}
\end{proposition}

Note that $\displaystyle\phi_{k}(z)$ is expressed as 
\begin{align}
\displaystyle\phi_{k}(z)=&\left(\frac{1}{iV\kappa}\right)^{k-1}\left(\frac{d}{dz}\right)^{k-1}\int_{-\infty}^{\infty}dx\exp\left(-i\frac{\lambda V}{3}x^3+iV\kappa xz\right).\label{b}
\end{align}

We use Airy function:
\begin{align}
Ai(\gamma L)=&\frac{1}{2\pi\gamma}\int_{-\infty}^{\infty}\exp\Biggl[i\left(Lx+\frac{x^{3}}{3\gamma^{3}}\right)\Biggl]dx.\label{a}
\end{align}
Here $\gamma\in\mathbb{R}\backslash\{0\}$ and $L\in\mathbb{R}$. Substituting (\ref{a}) for (\ref{b}), $\displaystyle\phi_{k}(z)$ is calculated as follows:
\begin{align}
\displaystyle\phi_{k}(z)=&\left(\frac{i}{(\lambda V)^{\frac{1}{3}}}\right)^{k-1}\left(\frac{-2\pi}{(\lambda V)^{\frac{1}{3}}}\right)\left(\frac{d}{dy}\right)^{k-1}Ai\left[y\right]\Biggl|_{y=-\frac{V\kappa z}{(\lambda V)^{\frac{1}{3}}}}.
\end{align}

\begin{proposition}(\cite{Kontsevich:1992ti}).

Let $\left(Ai^{(j-1)}(y_{i})\right)$ be the $N\times N$ matrix with the $i$-th row and the $j$-th column being $\displaystyle Ai^{(j-1)}(y_{i})=\left(\frac{d}{dy_{i}}\right)^{j-1}\!\!\!\!\!\!Ai(y_{i})$. We then obtain the following:
\begin{align}
\det\left(Ai^{(j-1)}(y_{i})\right)=&\left(\prod_{1\leq i<j\leq N}\left(\partial_{y_{i}}-\partial_{y_{j}}\right)\right)Ai(y_{1})\cdots Ai(y_{N}).\notag
\end{align}

The proof is omitted in \cite{Kontsevich:1992ti}, so it is appended for the reader's convenience.

\end{proposition}

\begin{proof}
We calculate $\displaystyle\det\left(Ai^{(j-1)}(y_{i})\right)$ according to the definition of the determinant.
\begin{align}
\det\left(Ai^{(j-1)}(y_{i})\right)=&\sum_{\sigma\in S_{N}}sgn\sigma\prod_{k=0}^{N-1}Ai^{(k)}(y_{\sigma(k+1)})\notag\\
=&\sum_{\sigma\in S_{N}}sgn\sigma\prod_{k=0}^{N-1}\partial_{y_{\sigma(k+1)}}^{k}Ai(y_{\sigma(k+1)}),\label{c}
\end{align}
where $S_{N}$ is the $N$-th order symmetry group. For $\displaystyle\prod_{k=0}^{N-1}\partial_{y_{\sigma(k+1)}}^{k}$ using similar calculation to the Vandermonde determinant, $\det\left(Ai^{(j-1)}(y_{i})\right)$ is as follows:
\begin{align}
\det\left(Ai^{(j-1)}(y_{i})\right)=&\sum_{\sigma\in S_{N}}sgn\sigma\left(\prod_{k=0}^{N-1}\partial_{y_{\sigma(k+1)}}^{k}\right)Ai(y_{1})\cdots Ai(y_{N})\notag\\
=&\left(\prod_{1\leq i<j\leq N}\left(\partial_{y_{i}}-\partial_{y_{j}}\right)\right)Ai(y_{1})\cdots Ai(y_{N}).
\end{align}

\end{proof}
We introduce 
\begin{align}
A_{N}(y_{1},\cdots,y_{N})=\left(\displaystyle\prod_{1\leq i< j\leq N}(\partial_{y_{i}}-\partial_{y_{j}})\right)Ai(y_{1})\cdots Ai(y_{N}),
\end{align}

where $\displaystyle y_{j}=-\frac{V\kappa s_{j}}{(\lambda V)^{\frac{1}{3}}}$ for $j=1,\ldots,N$. From this, $\displaystyle\det_{1\leq i,j\leq N}(\phi_{i}(s_{j}))$ is calculated as follows:
\begin{align}
\displaystyle\det_{1\leq i,j\leq N}(\phi_{i}(s_{j}))=&\left(\frac{(i)^{\frac{N(N-1)}{2}}(-2\pi)^{N}}{(\lambda V)^{\frac{N(N+1)}{6}}}\right)A_{N}(y_{1},\cdots,y_{N}).
\end{align}

Summarizing above results, we obtain the following:

\begin{proposition}

Let $\mathcal{Z}[J]$ be the partition function of the $\Phi_{2}^{3}$ matrix model given by (\ref{Z[J]}). Then, $\mathcal{Z}[J]$ is given as
\begin{align}
\mathcal{Z}[J]=&\int \mathcal{D}\Phi \exp\left(-iV\mathrm{tr}\left(E\Phi^2+\kappa\Phi+\frac{\lambda}{3}\Phi^3\right)\right)\exp\left(iV\mathrm{tr}\left(J\Phi\right)\right)\notag\\
=&C'\frac{e^{\frac{-iV}{\lambda}\mathrm{tr}(JE)}A_{N}(y_{1},\cdots,y_{N})}{\displaystyle\prod_{1\leq t<u\leq N}(s_{u}-s_{t})}.\label{ZZ}
\end{align}

Here $\displaystyle C'=\exp\left(-\frac{iV}{\lambda^{2}}\mathrm{tr}\left(\frac{2}{3}E^{3}-\lambda\kappa E\right)\right)\left(\prod_{p=1}^{N}p!\right)\frac{(-2)^{N}\pi^{\frac{N(N+1)}{2}}}{\lambda^{\frac{N(N+1)}{6}}V^{\frac{N(2N-1)}{3}}\kappa^{\frac{N(N-1)}{2}}}$, $s_{t}$ is the eigenvalues of the matrix $\displaystyle\frac{E^{2}}{\lambda\kappa}-I+\frac{J}{\kappa}$ for $t=1,\cdots,N$, and $\displaystyle y_{j}=-\frac{V\kappa s_{j}}{(\lambda V)^{\frac{1}{3}}}$ for $j=1,\cdots, N$.

\end{proposition}


\section{Calculation of $1$-Point Function $G_{|a|}$ and $2$-Point Functions $G_{|ab|}$}\label{sec4}
In the calculation of the 1-point function $G_{|a|}$, the external field $J$ can be treated as the diagonal matrix $J=diag\left(J_{11},\cdots,J_{NN}\right)$. Then the eigenvalues $s_{t}$ in (\ref{ZZ}) is given $\displaystyle s_{t}=\frac{\lambda(\widetilde{E}_{t})^{2}+J_{tt}}{\kappa}-1$. Then, the $1$-point function $G_{|a|}$ is calculated as follows:
\begin{align}
G_{|a|}=&\left.\frac{1}{iV}\frac{\partial\log\mathcal{Z}[J]}{\partial J_{aa}}\right|_{J=0}\notag\\
=&\displaystyle\frac{\displaystyle\left(\frac{1}{iV}\right)\displaystyle\frac{\partial}{\partial J_{aa}}\Biggl(\frac{e^{-iV\mathrm{tr}(\mathrm{J}\widetilde{E})}\displaystyle A_{N}(y_{1},\cdots,y_{N})}{\displaystyle\prod_{1\leq t<u\leq\mathrm{N}}\left(\frac{\lambda (\widetilde{E}_{u})^{2}-\lambda(\widetilde{E}_{t})^{2}+(J_{uu}-J_{tt})}{\kappa}\right)}\Biggl)\Biggl|_{J=0}}{\frac{\displaystyle A_{N}(y_{1},\cdots,y_{N})\Biggl|_{J=0}}{\displaystyle\prod_{1\leq p<q\leq\mathrm{N}}\left(\frac{\lambda}{\kappa}\{(\widetilde{E}_{q})^{2}- (\widetilde{E}_{p})^{2}\}\right)}},\label{d}
\end{align}
where $\displaystyle y_{j}=-\frac{VE^{2}_{j-1}}{(\lambda V)^{\frac{1}{3}}\lambda}+\frac{V\kappa}{(\lambda V)^{\frac{1}{3}}}-\frac{VJ_{jj}}{(\lambda V)^{\frac{1}{3}}}$ for $j=1,\ldots,N$. Note that
\begin{align}
&\frac{\partial}{\partial J_{aa}}\Biggl\{e^{-iV\mathrm{tr}(\mathrm{J}\widetilde{E})}\displaystyle A_{N}(y_{1},\cdots,y_{N})\Biggl\}\notag\\
&=-iV\widetilde{E}_{aa}e^{-iV\mathrm{tr}(\mathrm{J}\widetilde{E})}\displaystyle A_{N}(y_{1},\cdots,y_{N})+e^{-iV\mathrm{tr}(\mathrm{J}\widetilde{E})}\left(-\frac{V}{(\lambda V)^{\frac{1}{3}}}\right)\left(\partial_{a}\displaystyle A_{N}(y_{1},\cdots,y_{N})\right),\label{ykkkk}
\end{align}

where $\displaystyle\partial_{a}A_{N}(y_{1},\cdots,y_{N})=\frac{\partial}{\partial y_{a}}A_{N}(y_{1},\cdots,y_{N})$. Next, we use the following formula. Let $\displaystyle v_{n}=v_{n}(\vec{x}_{n})=\det_{1\leq i,j\leq N}(x_{j})^{i-1}$ be the Vandermonde determinant for $\vec{x}_{n}=(x_{1},\cdots,x_{n})\in\mathbb{R}^{n}$. For any $1\leq k\leq n$
\begin{align}
\frac{\partial v_{n}}{\partial x_{k}}=&\sum_{i=1,i\neq k}^{n}\frac{v_{n}(\vec{x}_{n})}{x_{k}-x_{i}}.\label{www}
\end{align}

(See for example \cite{KJS}.) Using this formula, we get
\begin{align}
&\frac{\partial}{\partial J_{aa}}\left\{\displaystyle\left(\frac{1}{\kappa}\right)^{\frac{N(N-1)}{2}}\det_{1\leq i,j\leq N}\Biggl(\lambda (\widetilde{E}_{j})^{2}+J_{jj}\Biggl)^{i-1}\right\}\notag\\
&=\displaystyle\left(\frac{1}{\kappa}\right)^{\frac{N(N-1)}{2}}\sum_{i=1,i\neq a}^{N}\frac{\displaystyle\det_{1\leq i,j\leq N}\left(\lambda(\widetilde{E}_{j})^{2}+J_{jj}\right)^{i-1}}{(\lambda(\widetilde{E}_{a})^{2}+J_{aa})-(\lambda(\widetilde{E}_{i})^{2}+J_{ii})}.\label{xkkkk}
\end{align}
Substituting (\ref{ykkkk}) and (\ref{xkkkk}) into (\ref{d}), finally $G_{|a|}$ is expressed as
\begin{align}
G_{|a|}=&-\frac{E_{a-1}}{\lambda}-\frac{\lambda}{iV}\sum_{i=1,i\neq a}^{N}\frac{1}{E_{a-1}^{2}-E_{i-1}^{2}}+\left(\frac{1}{i}\right)\left(-\frac{1}{(\lambda V)^{\frac{1}{3}}}\right)\partial_{a}\log A_{N}(z_{1},\cdots,z_{N}),\label{h}
\end{align}
where $\displaystyle z_{j}=-\frac{VE^{2}_{j-1}}{(\lambda V)^{\frac{1}{3}}\lambda}+\frac{V\kappa}{(\lambda V)^{\frac{1}{3}}}$ for $j=1,\ldots,N$, and $\displaystyle\partial_{a}=\frac{\partial}{\partial z_{a}}$.\\ \bigskip\bigskip

Next, let us consider $2$-point functions $G_{|ab|}$ ($a\neq b,\hspace{1mm} a,b\in\{1,2,\cdots,N\}$). For the calculation, we put $J$ as the matrix all components without $J_{ab},J_{ba}$ are zero. Note that $\mathrm{tr}JE=\mathrm{tr}J\widetilde{E}=0$ for this $J$. 

At first, we estimate eigenvalues $s_{t}$ for $t=1,\ldots,N$ of the matrix $M-I+K$. The eigenequation is 
\begin{align}
&0=\det\left(sI-(M-I+K)\right)\notag\\
&=\left(\prod_{i=1, i\neq a, i\neq b}^{N}\left(s-\frac{E_{i-1}^{2}}{\lambda\kappa}+1\right)\right)\Biggl\{s^{2}+\left(-\frac{E_{b-1}^{2}}{\lambda\kappa}-\frac{E_{a-1}^{2}}{\lambda\kappa}+2\right)s\notag\\
&\hspace{3mm}+\frac{E_{a-1}^{2}E_{b-1}^{2}}{\lambda^{2}\kappa^{2}}-\frac{E_{b-1}^{2}}{\lambda\kappa}-\frac{E_{a-1}^{2}}{\lambda\kappa}+1-\frac{J_{ab}J_{ba}}{\kappa^{2}}\Biggl\}.\notag\\
\end{align}
Eigenvalues of the matrix $M-I+K$ are labeled as $\displaystyle s_{t}=\frac{E_{t-1}^{2}}{\lambda\kappa}-1$ for $t\neq a,b$, 
\begin{align}
\displaystyle s_{a}=\frac{\displaystyle\frac{E_{a-1}^{2}}{\lambda\kappa}+\frac{E_{b-1}^{2}}{\lambda\kappa}-2+\sqrt{\left(\frac{E_{a-1}^{2}}{\lambda\kappa}-\frac{E_{b-1}^{2}}{\lambda\kappa}\right)^{2}+4\times\frac{J_{ab}J_{ba}}{\kappa^{2}}}}{2},\label{sa}
\end{align}
\begin{align}
s_{a}|_{J=0}=&\frac{E_{a-1}^{2}}{\lambda\kappa}-1
\end{align}

and 
\begin{align}
\displaystyle s_{b}=\frac{\displaystyle\frac{E_{a-1}^{2}}{\lambda\kappa}+\frac{E_{b-1}^{2}}{\lambda\kappa}-2-\sqrt{\left(\frac{E_{a-1}^{2}}{\lambda\kappa}-\frac{E_{b-1}^{2}}{\lambda\kappa}\right)^{2}+4\times\frac{J_{ab}J_{ba}}{\kappa^{2}}}}{2},\label{sb}
\end{align}
\begin{align}
s_{b}|_{J=0}=&\frac{E_{b-1}^{2}}{\lambda\kappa}-1.
\end{align}

Let us calculate $G_{|ab|}$ by using these $s_{t}$.
\begin{align}
G_{|ab|}=&\left.\frac{1}{iV}\frac{\partial^{2}\log\mathcal{Z}[J]}{\partial J_{ab}\partial J_{ba}}\right|_{J=0}\notag\\
=&\frac{1}{iV}\frac{\displaystyle\left.\frac{\displaystyle\frac{\displaystyle\partial^{2}}{\partial J_{ab}\partial J_{ba}}A_{N}(y_{1},\cdots,y_{N})}{\displaystyle\prod_{1\leq t<u\leq N}(s_{u}-s_{t})}\right|_{J=0}-\left.\frac{A_{N}(y_{1},\cdots,y_{N})\displaystyle\frac{\displaystyle\partial^{2}}{\partial J_{ab}\partial J_{ba}}\Biggl\{\prod_{1\leq t<u\leq N}(s_{u}-s_{t})\Biggl\}}{\displaystyle\Biggl\{\prod_{1\leq t<u\leq N}(s_{u}-s_{t})\Biggl\}^{2}}\right|_{J=0}}{\displaystyle\frac{A_{N}(y_{1},\cdots,y_{N})|_{J=0}}{\displaystyle\prod_{1\leq t<u\leq N}(s_{u}-s_{t})|_{J=0}}}.\label{e}
\end{align}

Here we use $\left.\displaystyle\frac{\partial A_{N}(y_{1},\cdots,y_{N})}{\partial J_{ab}}\right|_{J=0}=\displaystyle\left.\frac{\displaystyle\partial\det_{1\leq k,l\leq N}(s_{l})^{k-1}}{\partial J_{ab}}\right|_{J=0}=0$, since $s_{a}$ and $s_{b}$ are functions of $(J_{ab}J_{ba})$ as we see in (\ref{sa}) and (\ref{sb}), then $\displaystyle\frac{\partial A_{N}(y_{1},\cdots,y_{N})}{\partial J_{ab}}$ and $\displaystyle\frac{\displaystyle\partial\det_{1\leq k,l\leq N}(s_{l})^{k-1}}{\partial J_{ab}}$ are of the form $J_{ba}\times(\cdots)$. 

Recall that $\displaystyle y_{k}=-\frac{V\kappa s_{k}}{(\lambda V)^{\frac{1}{3}}}$. Using the fact that $\displaystyle\left.\frac{\partial y_{k}}{\partial J_{ab}}\right|_{J=0}=0$ and  $\displaystyle\frac{\partial^{2} y_{a}}{\partial J_{ba}\partial J_{ab}}=-\frac{V\lambda}{(\lambda V)^{\frac{1}{3}}}\frac{1}{E_{a-1}^{2}-E_{b-1}^{2}}=-\frac{\partial^{2}y_{b}}{\partial J_{ab}\partial J_{ba}}$, we obtain
\begin{align}
\displaystyle\left.\frac{\displaystyle\partial^{2}}{\partial J_{ab}\partial J_{ba}}A_{N}(y_{1},\cdots,y_{N})\right|_{J=0}=&\frac{V\lambda}{(\lambda V)^{\frac{1}{3}}}\frac{1}{E_{a-1}^{2}-E_{b-1}^{2}}\left(\partial_{b}A_{N}\left(z_{1},\ldots,z_{N}\right)-\partial_{a}A_{N}\left(z_{1},\ldots,z_{N}\right)\right),\label{rrr}
\end{align}
where $\displaystyle z_{j}=-\frac{VE^{2}_{j-1}}{(\lambda V)^{\frac{1}{3}}\lambda}+\frac{V\kappa}{(\lambda V)^{\frac{1}{3}}}$ for $j=1,\ldots,N$. Similarly, we get 
\begin{align}
&\displaystyle\left.\frac{\displaystyle\partial^{2}}{\partial J_{ab}\partial J_{ba}}\Biggl\{\prod_{1\leq t<u\leq N}(s_{u}-s_{t})\Biggl\}\right|_{J=0}\notag\\
&=\frac{\lambda^{2}}{E_{a-1}^{2}-E_{b-1}^{2}}\left(\sum_{i=1,i\neq a}^{N}\frac{1}{E_{a-1}^{2}-E_{i-1}^{2}}-\sum_{i=1,i\neq b}^{N}\frac{1}{E_{b-1}^{2}-E_{i-1}^{2}}\right)\det_{1\leq k,l\leq N}(s_{k})^{l-1},\label{rrrr}
\end{align}
where we use the formula (\ref{www}), again. Substituting (\ref{rrr}) and (\ref{rrrr}) into (\ref{e}), $G_{|ab|}$\hspace{2mm}($b<a$, and $E_{b}<E_{a}$) is finally obtained as
\begin{align}
G_{|ab|}=&-\frac{\lambda^{2}}{iV}\sum_{i=1,i\neq a}^{N}\frac{1}{(E_{a-1}^{2}-E_{i-1}^{2})(E_{a-1}^{2}-E_{b-1}^{2})}+\frac{\lambda^{2}}{iV}\sum_{i=1,i\neq b}^{N}\frac{1}{(E_{b-1}^{2}-E_{i-1}^{2})(E_{a-1}^{2}-E_{b-1}^{2})}\notag\\
&-\frac{\lambda}{i(\lambda V)^{\frac{1}{3}}}\frac{1}{E_{a-1}^{2}-E_{b-1}^{2}}\partial_{a}\log A_{N}(z_{1},\ldots,z_{N})+\frac{\lambda}{i(\lambda V)^{\frac{1}{3}}}\frac{1}{E_{a-1}^{2}-E_{b-1}^{2}}\partial_{b}\log A_{N}(z_{1},\ldots,z_{N}).\label{f}
\end{align}
We now refer to the Schwinger-Dyson equation
\begin{align}
G_{|ab|}=&\frac{1}{E_{a-1}+E_{b-1}}\left(1+\lambda\frac{\left(G_{|a|}-G_{|b|}\right)}{E_{a-1}-E_{b-1}}\right)\label{g}
\end{align}
in reference\cite{Grosse:2016pob}. Substituting (\ref{f}) for the left side of (\ref{g}) and (\ref{h}) for the right side of (\ref{g}) shows that Schwinger-Dyson equation (\ref{g}) is indeed satisfied.


\section{Calculation of $n$-Point Functions $G_{|a^{1}|a^{2}|\cdots|a^{n}|}$}\label{sec5}
The goal of this section is to obtain the explicit formula of the n-point function $G_{|a^{1}|a^{2}|\cdots|a^{n}|}$. Here $a^{\beta}$ is the pairwise different indices for $\beta=1,\ldots,n$. From the definition in (\ref{logZ}), the n-point function $G_{|a^{1}|a^{2}|\cdots|a^{n}|}$ is given by
\begin{align}
G_{|a^{1}|a^{2}|\cdots|a^{n}|}=&(iV)^{n-2}\frac{\partial^{n}}{\partial J_{a^{1}a^{1}}\ldots\partial J_{a^{n}a^{n}}}\left.\log\frac{\mathcal{Z}[J]}{\mathcal{Z}[0]}\right|_{J=0}.\label{G}
\end{align}

We use the formula in \cite{Hardy}:
\begin{align}
\frac{\partial^{n}}{\partial x_{1}\cdots\partial x_{n}}f(y)=&\sum_{\pi}f^{|\pi|}(y)\prod_{B\in\pi}\frac{\partial^{|B|}y}{\displaystyle\prod_{j\in B}\partial x_{j}},\label{kk}
\end{align}

where $f(y)$ is the differentiable function of the variable $y=y(x_{1},x_{2},\ldots,x_{n})$, $\displaystyle\sum_{\pi}$ means the sum over all partitions $\pi$ of the set $\{1,\ldots,n\}$, $\displaystyle\prod_{B\in\pi}$ is the product over all of the parts $B$ of the partition $\pi$, and $|S|$ denotes the cardinality of any set $S$. Applying (\ref{kk}) to (\ref{G}) n-point functions $G_{|a^{1}|a^{2}|\cdots|a^{n}|}$ is expressed as follows:
\begin{align}
G_{|a^{1}|a^{2}|\cdots|a^{n}|}=&(iV)^{n-2}\sum_{\pi}\Biggl\{\left.\left(\frac{d}{dx}\right)^{|\pi|}(\log x)\right|_{x=\mathcal{Z}[0]}\Biggl\}\prod_{B\in \pi}\frac{\partial^{|B|}\mathcal{Z}[J]}{\displaystyle\prod_{j\in B}\partial J_{a^{j}a^{j}}}\biggl|_{J=0}.\label{i}
\end{align}

After the calculation of $\displaystyle\frac{\partial^{|B|}\mathcal{Z}[J]}{\displaystyle\prod_{j\in B}\partial J_{a^{j}a^{j}}}\biggl|_{J=0}$, we get the following result.

\begin{lemma}\label{pxq}

We introduce $\displaystyle\mathcal{C}=\exp\left(-\frac{iV}{\lambda^{2}}\mathrm{tr}\left(\frac{2}{3}E^{3}-\lambda\kappa E\right)\right)\left(\prod_{p=1}^{N}p!\right)(-2)^{N}\frac{\pi^{\frac{N(N+1)}{2}}}{V^{\frac{N(2N-1)}{3}}\lambda^{\frac{N(N+1)}{6}}}$. Then 
\begin{align}
\frac{\partial^{|B|}\mathcal{Z}[J]}{\displaystyle\prod_{j\in B}\partial J_{a^{j}a^{j}}}\biggl|_{J=0}=&\mathcal{C}\sum_{S\subset B}\left(\left(\prod_{i\in S}\left(-iV\frac{E_{a^{i}-1}}{\lambda}\right)\right)\sum_{M\subset\overline{S}}\left(\left(\left\{\prod_{k\in M}\left(-\frac{V}{(\lambda V)^{\frac{1}{3}}}\right)\partial_{a^{k}}\right\}A_{N}(z_{1},\ldots,z_{N})\right)\right.\right.\notag\\
&\left.\left.\left(\left\{\prod_{q\in\overline{M}}\frac{\partial}{\displaystyle\partial t_{a^{q}}}\right\}\displaystyle\frac{1}{\displaystyle\det_{1\leq l,j\leq N}\left(t_{l}^{j-1}\right)}\right)\right)\right),\label{A}
\end{align}
where $\displaystyle z_{j}=-\frac{VE_{j-1}^{2}}{(\lambda V)^{\frac{1}{3}}\lambda}+\frac{V\kappa}{(\lambda V)^{\frac{1}{3}}}$ for $j=1,\ldots N$, $\displaystyle\partial_{a^{k}}=\frac{\partial}{\partial z_{a^{k}}}\hspace{1mm}(k\in M)$, $\displaystyle t_{l}=\displaystyle\frac{(E_{l-1})^{2}}{\lambda}$ for $l=1,\ldots N$, $S$ runs through the set of all subsets of $B$, $\overline{S}$ is the complement of $S$ in $B$, $M$ runs through the set of all subsets of $\overline{S}$, and $\overline{M}=\overline{S}\backslash M$.

\end{lemma}

\begin{proof}

For the calculation of $G_{|a^{1}|a^{2}|\cdots|a^{n}|}$, we can choose $J$ as a diagonal matrix $diag(J_{11},\cdots,J_{NN})$. Then, $\displaystyle s_{t}=\frac{\lambda(\widetilde{E}_{t})^{2}+J_{tt}}{\kappa}-1$. To calculate 
\begin{align}
\frac{\partial^{|B|}\mathcal{Z}[J]}{\displaystyle\prod_{j\in B}\partial J_{a^{j}a^{j}}}\biggl|_{J=0}=&\mathcal{C}\frac{\partial^{|B|}}{\displaystyle\prod_{j\in B}\partial J_{a^{j}a^{j}}}\biggl(\frac{e^{-iVtr(J\widetilde{E})}A_{N}(y_{1},\cdots,y_{N})}{\displaystyle\det_{1\leq i,j\leq N}\left(\lambda(\widetilde{E}_{j})^{2}+J_{jj}\right)^{i-1}}\biggl)\biggl|_{J=0}\label{kkk}
\end{align}
we use the formula in \cite{Hardy}:
\begin{align}
\frac{\partial^{n}}{\partial x_{1}\cdots\partial x_{n}}(uv)=&\sum_{S}\frac{\partial^{|S|}u}{\prod_{j\in S}\partial x_{j}}\cdot\frac{\partial^{(n-|S|)}v}{\prod_{j\notin S}\partial x_{j}},\label{xxa}
\end{align}
where $u$, and $v$ are differentiable functions of the variable $x=(x_{1},x_{2},\ldots,x_{n})$, and $S$ runs through the set of all subsets of $\{1,\ldots,n\}$.

Using the formula (\ref{xxa}) twice for (\ref{kkk}), we obtain the following :
\begin{align}
\frac{\partial^{|B|}\mathcal{Z}[J]}{\displaystyle\prod_{j\in B}\partial J_{a^{j}a^{j}}}\biggl|_{J=0}=&\mathcal{C}\sum_{S\subset B}\frac{\partial^{|S|}e^{-iVtr(J\widetilde{E})}}{\displaystyle\prod_{i\in S}\partial J_{a^{i}a^{i}}}\biggl|_{J=0}\sum_{M\subset \overline{S}}\frac{\partial^{|M|}A_{N}(y_{1},\cdots,y_{N})}{\displaystyle\prod_{k\in M}\partial J_{a^{k}a^{k}}}\biggl|_{J=0}\notag\\
&\frac{\partial^{|\overline{M}|}}{\displaystyle\prod_{q\in\overline{M}}\partial J_{a^{q}a^{q}}}\biggl(\frac{1}{\displaystyle\det_{1\leq l,j\leq N}\left(\lambda(\widetilde{E}_{l})^{2}+J_{ll}\right)^{j-1}}\biggl)\biggl|_{J=0}.\label{yas}
\end{align}

For the diagonal $J$, $\displaystyle y_{k}=\left(-\frac{V\kappa}{(\lambda V)^{\frac{1}{3}}}\right)\left(\frac{E_{k-1}^{2}}{\lambda\kappa}-1+\frac{J_{kk}}{\kappa}\right)$, then the above is rewritten as
\begin{align}
(\ref{yas})=& \mathcal{C}\sum_{S\subset B}\left(\left(\prod_{i\in S}\left(-iV\frac{E_{a^{i}-1}}{\lambda}\right)\right)\sum_{M\subset\overline{S}}\left(\left(\left\{\prod_{k\in M}\left(-\frac{V}{(\lambda V)^{\frac{1}{3}}}\right)\partial_{a^{k}}\right\}A_{N}(z_{1},\ldots,z_{N})\right)\right.\right.\notag\\
&\times\left.\left.\left(\left\{\prod_{q\in\overline{M}}\frac{\partial}{\displaystyle\partial t_{a^{q}}}\right\}\displaystyle\frac{1}{\displaystyle\det_{1\leq l,j\leq N}\left(t_{l}^{j-1}\right)}\right)\right)\right),
\end{align}
where $\displaystyle z_{j}=-\frac{VE_{j-1}^{2}}{(\lambda V)^{\frac{1}{3}}\lambda}+\frac{V\kappa}{(\lambda V)^{\frac{1}{3}}}$ for $j=1,\ldots N$, $\displaystyle t_{l}=\displaystyle\frac{(E_{l-1})^{2}}{\lambda}$ for $l=1,\ldots N$, $S$ runs through the set of all subsets of $B$, $\overline{S}$ is the complement of $S$ in $B$, $M$ runs through the set of all subsets of $\overline{S}$, and $\overline{M}=\overline{S}\backslash M$.

\end{proof}

Note the cases that each set is an empty set, 
\begin{align*}
&\prod_{i\in \emptyset}\left(-iV\frac{E_{a^{i}-1}}{\lambda}\right)=1, \qquad
\left\{\displaystyle\prod_{k\in \emptyset}\left(-\frac{V}{(\lambda V)^{\frac{1}{3}}}\right)\partial_{a^{k}}\right\}A_{N}(z_{1},\ldots,z_{N})=A_{N}(z_{1},\ldots,z_{N}),\\
&\mbox{and }~
\left\{\prod_{q\in\emptyset}\frac{\partial}{\displaystyle\partial t_{a^{q}}}\right\}\displaystyle\frac{1}{\displaystyle\det_{1\leq l,j\leq N}\left(t_{l}^{j-1}\right)}=\frac{1}{\displaystyle\det_{1\leq l,j\leq N}\left(t_{l}^{j-1}\right)}.
\end{align*}

Summarizing (\ref{i}) and the result in Lemma \ref{pxq}, we obtain the following:
\begin{theorem}\label{thm5.2}
We suppose the partition function $\mathcal{Z}[J]$ of the $\Phi_{2}^{3}$ matrix model is defined by (\ref{Z[J]}) and $G_{|a^{1}|a^{2}|\cdots|a^{n}|}$ is defined by (\ref{G}). In this case,
\begin{align}
&G_{|a^{1}|a^{2}|\cdots|a^{n}|}\notag\\
=&(iV)^{n-2}\mathcal{C}\sum_{\pi}\Biggl\{\left.\left(\frac{d}{dx}\right)^{|\pi|}(\log x)\right|_{x=\mathcal{Z}[0]}\Biggl\}\prod_{B\in \pi}\sum_{S\subset B}\left(\left(\prod_{i\in S}\left(-iV\frac{E_{a^{i}-1}}{\lambda}\right)\right)\right.\notag\\
&\times\left.\sum_{M\subset\overline{S}}\left(\left(\left\{\prod_{k\in M}\left(-\frac{V}{(\lambda V)^{\frac{1}{3}}}\right)\partial_{a^{k}}\right\}A_{N}(z_{1},\ldots,z_{N})\right)\left(\left\{\prod_{q\in\overline{M}}\frac{\partial}{\displaystyle\partial t_{a^{q}}}\right\}\displaystyle\frac{1}{\displaystyle\prod_{1\leq l<j\leq N}(t_{j}-t_{l})}\right)\right)\right),\label{K}
\end{align}
where $\displaystyle\sum_{\pi}$ means the sum over all partitions $\pi$ of the set $\{1,\ldots,n\}$, $\displaystyle\prod_{B\in\pi}$ is over all of the parts $B$ of the partition $\pi$, $|S|$ denotes the cardinality of any set $S$, $\displaystyle\sum_{S\subset B}$ means the sum over all subsets of $B$, $\displaystyle\sum_{M\subset\overline{S}}$ means the sum over all subsets of $\overline{S}=B\backslash S$, and $t_{l}=\displaystyle\frac{(E_{l-1})^{2}}{\lambda}$.
\end{theorem}

Now we refer to the formula in Section 5 in \cite{Grosse:2016pob}.

\begin{theorem}(\cite{Grosse:2016pob}).\label{thm5.3}
We suppose $G_{|a^{1}|a^{2}|\cdots|a^{n}|}$ is defined by (\ref{G}). In this case,
\begin{align}
&G_{|a_{1}^{1}\ldots a_{N_{1}}^{1}|\ldots|a_{1}^{B}\ldots a_{N_{B}}^{B}|}\notag\\
=&\lambda^{N_{1}+\cdots+N_{B}-B}\sum_{k_{1}=1}^{N_{1}}\ldots\sum_{k_{B}=1}^{N_{B}}G_{|a^{1}_{k_{1}}|\ldots|a^{B}_{k_{B}}|}\left(\prod_{l_{1}=1,l_{1}\neq k_{1}}^{N_{1}}P_{a_{k_{1}}^{1}a_{l_{1}}^{1}}\right)\cdots\left(\prod_{l_{B}=1,l_{B}\neq k_{B}}^{N_{B}}P_{a_{k_{B}}^{B}a_{l_{B}}^{B}}\right)\label{ab},
\end{align}
where $2\leq B$, $N_{i}>1$ for $i=1,\ldots,B$, $\lambda$ is the coupling constant(real), and $\displaystyle P_{ab}:=\frac{1}{E_{a-1}^{2}-E_{b-1}^{2}}$.

\end{theorem}
Substituting (\ref{K}) into (\ref{ab}), all the exact solutions of the $\Phi_{2}^{3}$ finite matrix model is obtained.\\\bigskip

For the later convenience, we introduce a function $F(S,M,\overline{M})$. Let $B$ be a subset of $\{1,\cdots,n\}$. For $S\subset B$, $\overline{S}$ denotes the complement $B\backslash S$.
\begin{align}
&F(S,M,\overline{M})\notag\\
:=&\left(\prod_{i\in S}\left(-iV\frac{E_{a^{i}-1}}{\lambda}\right)\right)\left(\left\{\prod_{k\in M}\left(-\frac{V}{(\lambda V)^{\frac{1}{3}}}\right)\partial_{a^{k}}\right\}A_{N}(z_{1},\ldots,z_{N})\right)\left(\left\{\prod_{q\in\overline{M}}\frac{\partial}{\displaystyle\partial t_{a^{q}}}\right\}\displaystyle\frac{1}{\displaystyle\det_{1\leq l,j\leq N}\left(t_{l}^{j-1}\right)}\right),
\end{align}
where $B=S\sqcup\overline{S}$, $\overline{S}=M\sqcup\overline{M}$, and $\displaystyle\partial_{a^{k}}=\frac{\partial}{\partial z_{a^{k}}}$. Using this $F(S,M,\overline{M})$, $G_{|a^{1}|a^{2}|\cdots|a^{n}|}$ is expressed as 
\begin{align}
G_{|a^{1}|a^{2}|\cdots|a^{n}|}=&(iV)^{n-2}\mathcal{C}\sum_{\pi}\Biggl\{\left.\left(\frac{d}{dx}\right)^{|\pi|}(\log x)\right|_{x=\mathcal{Z}[0]}\Biggl\}\prod_{B\in \pi}\sum_{S\subset B}\sum_{M\subset\overline{S}} F(S,M,\overline{M}).\label{Fs}
\end{align}


\section{Calculation of Two-Point Functions $G_{|a|b|}$}\label{sec6}

The formula (\ref{Fs}) is used to obtain $G_{|a|b|}$ concretely. We use $a=a^{1}$ and $b=a^{2}$ below. \\
At first we estimate the case of $\pi=\{\{1,2\}\}$. In this case $|\pi|=1$, and it is enough to calculate $F(S,M,\overline{M})$ for $B=\{1,2\}$.
In the context of Theorem \ref{thm5.2}, it corresponds to the part: 
\begin{align}
\left.\left(\frac{d}{dx}\right)^{|\pi|}(\log x)\right|_{x=\mathcal{Z}[0]}\prod_{B\in \pi}\frac{\partial^{|B|}\mathcal{Z}[J]}{\displaystyle\prod_{j\in B}\partial J_{a^{j}a^{j}}}\biggl|_{J=0}=&\left.\frac{1}{\mathcal{Z}[0]}\frac{\partial^{2}\mathcal{Z}[J]}{\partial J_{aa}\partial J_{bb}}\right|_{J=0}.\label{l}
\end{align}

Calculating all cases for sets $S$, $M$, and $\overline{M}$, we obtain the following results. In the case of $F(\{1,2\},\emptyset,\emptyset)$, 
\begin{align}
F(\{1,2\},\emptyset,\emptyset)=&\left(-iV\frac{E_{a-1}}{\lambda}\right)\left(-iV\frac{E_{b-1}}{\lambda}\right)A_{N}(z_{1},\ldots,z_{N})\frac{1}{\displaystyle\det_{1\leq l,j\leq N}\left(t_{l}^{j-1}\right)}.
\end{align}

In the case of $F(\{1\},\{2\},\emptyset)$, 
\begin{align}
F(\{1\},\{2\},\emptyset)=&\left(-iV\frac{E_{a-1}}{\lambda}\right)\left(-\frac{V}{(\lambda V)^{\frac{1}{3}}}\right)\partial_{b}A_{N}(z_{1},\ldots,z_{N})\frac{1}{\displaystyle\det_{1\leq l,j\leq N}\left(t_{l}^{j-1}\right)}.\label{tttt}
\end{align}

$F(\{2\},\{1\},\emptyset)$ can be calculated in the same way (\ref{tttt}). The letters $a$ and $b$ in (\ref{tttt}) are interchanged. In the case of $F(\emptyset,\{1,2\},\emptyset)$, 
\begin{align}
F(\emptyset,\{1,2\},\emptyset)=&\left(-\frac{V}{(\lambda V)^{\frac{1}{3}}}\right)^{2}\partial_{a}\partial_{b}A_{N}(z_{1},\ldots,z_{N})\frac{1}{\displaystyle\det_{1\leq l,j\leq N}\left(t_{l}^{j-1}\right)}.
\end{align}

In the case of $F(\{1\},\emptyset,\{2\})$, 
\begin{align}
F(\{1\},\emptyset,\{2\})=&\left(-iV\frac{E_{a-1}}{\lambda}\right)A_{N}(z_{1},\ldots,z_{N})\frac{-1}{\displaystyle\det_{1\leq l,j\leq N}\left(t_{l}^{j-1}\right)}\sum_{i=1,i\neq b}^{N}\frac{1}{t_{b}-t_{i}}.\label{ttt}
\end{align}

$F(\{2\},\emptyset,\{1\})$ can be calculated in the same way (\ref{ttt}). The letters $a$ and $b$ in (\ref{ttt}) are interchanged.  In the case of $F(\emptyset,\{1\},\{2\})$, 
\begin{align}
F(\emptyset,\{1\},\{2\})=&\left(-\frac{V}{(\lambda V)^{\frac{1}{3}}}\right)\partial_{a}A_{N}(z_{1},\ldots,z_{N})\frac{-1}{\displaystyle\det_{1\leq l,j\leq N}\left(t_{l}^{j-1}\right)}\sum_{i=1,i\neq b}^{N}\frac{1}{t_{b}-t_{i}}.\label{ttttttt}
\end{align}

$F(\emptyset,\{2\},\{1\})$ can be calculated in the same way (\ref{ttttttt}). The letters $a$ and $b$ in (\ref{ttttttt}) are interchanged. In the case of $F(\emptyset,\emptyset,\{1,2\})$, 
\begin{align}
&F(\emptyset,\emptyset,\{1,2\})\notag\\
=&A_{N}(z_{1},\ldots,z_{N})\frac{1}{\displaystyle\det_{1\leq l,j\leq N}\left(t_{l}^{j-1}\right)}\left(\sum_{i=1,i\neq a}^{N}\frac{1}{t_{a}-t_{i}}\sum_{j=1,j\neq b}^{N}\frac{1}{t_{b}-t_{j}}-\frac{1}{(t_{a}-t_{b})^{2}}\right).
\end{align}
From this, (\ref{l}) can be calculated as follows:
\begin{align}
&\left.\frac{1}{\mathcal{Z}[0]}\frac{\partial^{2}\mathcal{Z}[J]}{\partial J_{aa}\partial J_{bb}}\right|_{J=0}\notag\\
=&\left(\frac{\displaystyle\det_{1\leq l,j\leq N}\left(t_{l}^{j-1}\right)}{A_{N}(z_{1},\ldots,z_{N})}\right)\Biggl\{F(\{1,2\},\emptyset,\emptyset)+F(\emptyset,\{1,2\},\emptyset)+F(\emptyset,\emptyset,\{1,2\})\notag\\
&+\sum_{l,n=1,l\neq n}^{2}\Biggl(F(\{l\},\{n\},\emptyset)+F(\{l\},\emptyset,\{n\})+F(\emptyset,\{l\},\{n\})\Biggl)\Biggl\}\label{ppp}
\end{align}

Next step, let us consider the case
$\pi=\{\{1\},\{2\}\}$, $|\pi|=2$, $B=\{1\}, or \{2\}$. 

The corresponding term  $\displaystyle\left.\left(\frac{d}{dx}\right)^{|\pi|}(\log x)\right|_{x=\mathcal{Z}[0]}\prod_{B\in \pi}\frac{\partial^{|B|}\mathcal{Z}[J]}{\displaystyle\prod_{j\in B}\partial J_{a^{j}a^{j}}}\biggl|_{J=0}$ in Theorem \ref{thm5.2} is as follows:
\begin{align}
\left.\left(\frac{d}{dx}\right)^{|\pi|}(\log x)\right|_{x=\mathcal{Z}[0]}\prod_{B\in \pi}\frac{\partial^{|B|}\mathcal{Z}[J]}{\displaystyle\prod_{j\in B}\partial J_{a^{j}a^{j}}}\biggl|_{J=0}=&-\left.\frac{1}{\mathcal{Z}[0]^{2}}\frac{\partial\mathcal{Z}[J]}{\partial J_{aa}}\right|_{J=0}\left.\frac{\partial\mathcal{Z}[J]}{\partial J_{bb}}\right|_{J=0}.\label{o}
\end{align}

Calculating all cases for sets $S$, $M$, and $\overline{M}$ of $B=\{1\}$, we obtain the following results. 
\begin{align}
F(\{1\},\emptyset,\emptyset)=&-iV\frac{E_{a-1}}{\lambda}A_{N}(z_{1},\ldots,z_{N})\frac{1}{\displaystyle\det_{1\leq l,j\leq N}\left(t_{l}^{j-1}\right)}.
\end{align}
\begin{align}
F(\emptyset,\{1\},\emptyset)=&\left(-\frac{V}{(\lambda V)^{\frac{1}{3}}}\right)\partial_{a}A_{N}(z_{1},\ldots,z_{N})\frac{1}{\displaystyle\det_{1\leq l,j\leq N}\left(t_{l}^{j-1}\right)}.
\end{align}
\begin{align}
F(\emptyset,\emptyset,\{1\})=&A_{N}(z_{1},\ldots,z_{N})\frac{-1}{\displaystyle\det_{1\leq l,j\leq N}\left(t_{l}^{j-1}\right)}\sum_{i=1,i\neq a}^{N}\frac{1}{t_{a}-t_{i}}.
\end{align}

These results can be summarized as follows:
\begin{align}
\left.\frac{1}{\mathcal{Z}[0]}\frac{\partial\mathcal{Z}[J]}{\partial J_{aa}}\right|_{J=0}=&\left(\frac{\displaystyle\det_{1\leq l,j\leq N}\left(t_{l}^{j-1}\right)}{A_{N}(z_{1},\ldots,z_{N})}\right)\Biggl\{F(\{1\},\emptyset,\emptyset)+F(\emptyset,\{1\},\emptyset)+F(\emptyset,\emptyset,\{1\})\Biggl\}.\label{m}
\end{align}

The same calculation is performed for $B=\{2\}$ as for $B=\{1\}$ :
\begin{align}
\left.\frac{1}{\mathcal{Z}[0]}\frac{\partial\mathcal{Z}[J]}{\partial J_{bb}}\right|_{J=0}=&\left(\frac{\displaystyle\det_{1\leq l,j\leq N}\left(t_{l}^{j-1}\right)}{A_{N}(z_{1},\ldots,z_{N})}\right)\Biggl\{F(\{2\},\emptyset,\emptyset)+F(\emptyset,\{2\},\emptyset)+F(\emptyset,\emptyset,\{2\})\Biggl\}.\label{n}
\end{align}

Note that (\ref{m}) and (\ref{n}) coincide with $iV$ multiples of the one-point function $G_{|a|}$ and $G_{|b|}$ in Section 3.
Substituting (\ref{m}) and (\ref{n}) into (\ref{o}) gives the result :
\begin{align}
&-\left.\frac{1}{\mathcal{Z}[0]^{2}}\frac{\partial\mathcal{Z}[J]}{\partial J_{aa}}\right|_{J=0}\left.\frac{\partial\mathcal{Z}[J]}{\partial J_{bb}}\right|_{J=0}\notag\\
=&-\left(\frac{\displaystyle\det_{1\leq l,j\leq N}\left(t_{l}^{j-1}\right)}{A_{N}(z_{1},\ldots,z_{N})}\right)^{2}\prod_{l=1}^{2}\Biggl(F(\{l\},\emptyset,\emptyset)+F(\emptyset,\{l\},\emptyset)+F(\emptyset,\emptyset,\{l\})\Biggl).\label{pppp}
\end{align}

Finally, adding (\ref{ppp}) and (\ref{pppp}) the result of the two point functions $G_{|a|b|}$ is obtained by
\begin{align}
G_{|a|b|}=&\left.\frac{1}{\mathcal{Z}[0]}\frac{\partial^{2}\mathcal{Z}[J]}{\partial J_{aa}\partial J_{bb}}\right|_{J=0}-\left.\frac{1}{\mathcal{Z}[0]^{2}}\frac{\partial\mathcal{Z}[J]}{\partial J_{aa}}\right|_{J=0}\left.\frac{\partial\mathcal{Z}[J]}{\partial J_{bb}}\right|_{J=0}\notag\\
=&\left(\frac{\displaystyle\det_{1\leq l,j\leq N}\left(t_{l}^{j-1}\right)}{A_{N}(z_{1},\ldots,z_{N})}\right)\Biggl\{F(\{1,2\},\emptyset,\emptyset)+F(\emptyset,\{1,2\},\emptyset)+F(\emptyset,\emptyset,\{1,2\})\notag\\
&+\sum_{l,n=1,l\neq n}^{2}\Biggl(F(\{l\},\{n\},\emptyset)+F(\{l\},\emptyset,\{n\})+F(\emptyset,\{l\},\{n\})\Biggl)\Biggl\}\notag\\
&-\left(\frac{\displaystyle\det_{1\leq l,j\leq N}\left(t_{l}^{j-1}\right)}{A_{N}(z_{1},\ldots,z_{N})}\right)^{2}\prod_{l=1}^{2}\Biggl(F(\{l\},\emptyset,\emptyset)+F(\emptyset,\{l\},\emptyset)+F(\emptyset,\emptyset,\{l\})\Biggl).
\end{align}

For a more complex example, we carry out the calculation for $G_{|a^{1}|a^{2}|a^{3}|}$ in Appendix \ref{Appendix}.


\section{Summary}\label{sec7}

In this paper, we found the exact solutions of the $\Phi_{2}^{3}$ finite matrix model (Grosse-Wulkenhaar model).
In the $\Phi_{2}^{3}$ finite matrix model, multipoint correlation functions were expressed as $G_{|a_{1}^{1}\ldots a_{N_{1}}^{1}|\ldots|a_{1}^{B}\ldots a_{N_{B}}^{B}|}$.
It is known that any $G_{|a_{1}^{1}\ldots a_{N_{1}}^{1}|\ldots|a_{1}^{B}\ldots a_{N_{B}}^{B}|}$ can be expressed using $G_{|a^{1}|\ldots|a^{n}|}$ type $n$-point functions as (\ref{ab}). 
Thus we focused on rigorous calculations of $G_{|a^{1}|\ldots|a^{n}|}$.

In Section \ref{sec3}, the integration of the off-diagonal elements of the Hermitian matrix was calculated using the Harish-Chandra-Itzykson-Zuber integral\cite{Itzykson:1979fi},\cite{T.Tao},\cite{Harish-Chandra:1957dhy},\cite{Zinn-Justin:2002rai} in calculating the partition function $\mathcal{Z}[J]$.
Next the integral of the diagonal elements of the Hermitian matrix was calculated using the Airy functions as similar to \cite{Kontsevich:1992ti}.
In Section \ref{sec4}, \ref{sec5}, and \ref{sec6}, we used the obtained partition function $\mathcal{Z}[J]$ to calculate $G_{|a^{1}|\ldots|a^{n}|}$ type $n$-point functions and $G_{|ab|}$. 
The exact solutions of $G_{|a^{1}|\ldots|a^{n}|}$ type $n$-point functions were found by calculating the $n$-th derivative $\partial^{n}/\partial J_{a^{1}a^{1}}\cdots\partial J_{a^{n}a^{n}}$ of  $\log\mathcal{Z}[J]$ with the external field $J$ as a diagonal matrix.
The result of the calculations for $G_{|a^{1}|\ldots|a^{n}|}$ was described in Theorem \ref{thm5.2}.
In the formula for $G_{|a^{1}|\ldots|a^{n}|}$ in Theorem \ref{thm5.2}, no integral remains.
More concletely, the $n$-point function was determined by a function $F(S,M,\overline{M})$ whose variables are $S\subset B$, $M$, and $\overline{M}$ $\left(B\backslash S=M\sqcup\overline{M}\right)$ as formula (\ref{Fs}), where $B$ is an element of a partition of $\{1,\cdots,n\}$.
Since the algorithm for finding the exact solutions of $G_{|a^{1}|\ldots|a^{n}|}$ type $n$-point functions is explicitly determined in the formula of Theorem \ref{thm5.2}, the exact solutions can be obtained automatically. Indeed, the calculations for $G_{|a|b|}$ and $G_{|a|b|c|}$ were carried out in Section \ref{sec6} and Appendix \ref{Appendix}, respectively.
Since a general $(N_{1}+\cdots +N_{B})$-point function $G_{|a_{1}^{1}\ldots a_{N_{1}}^{1}|\ldots|a_{1}^{B}\ldots a_{N_{B}}^{B}|}$ is expressed by using $G_{|a^{1}|\ldots|a^{n}|}$ type $n$-point functions, we can obtain all the exact solutions of the $\Phi_{2}^{3}$ finite matrix model.

%
%
\section*{Acknowledgements}

Authors are grateful to Professor H. Itou for his helpful advice. We thank D. Prekrat for pointing out an error in an equation. A.S. was supported by JSPS KAKENHI Grant Number 21K03258.


\appendix
\section{Calculation of the three-point functions $G_{|a^{1}|a^{2}|a^{3}|}$}\label{Appendix}

We calculate the three point functions $G_{|a^{1}|a^{2}|a^{3}|}$ using the formula (\ref{K}) or (\ref{Fs}). $i,l,k\in\{1,2,3\}$ and $i\neq l\neq k\neq i$ below.

\begin{enumerate}
\renewcommand{\labelenumi}{\roman{enumi}).}

 \item We consider the case $\pi=\{\{1,2,3\}\}$,$|\pi|=1$, and $B=\{1,2,3\}$, then 
\begin{align}
\left.\left(\frac{d}{dx}\right)^{|\pi|}(\log x)\right|_{x=\mathcal{Z}[0]}\prod_{B\in \pi}\frac{\partial^{|B|}\mathcal{Z}[J]}{\displaystyle\prod_{j\in B}\partial J_{a^{j}a^{j}}}\biggl|_{J=0}=&\left.\frac{1}{\mathcal{Z}[0]}\frac{\partial^{3}\mathcal{Z}[J]}{\partial J_{a^{1}a^{1}}\partial J_{a^{2}a^{2}}\partial J_{a^{3}a^{3}}}\right|_{J=0}.\label{v}
\end{align}

The calculations required to calculate $\displaystyle\left.\frac{1}{\mathcal{Z}[0]}\frac{\partial^{3}\mathcal{Z}[J]}{\partial J_{a^{1}a^{1}}\partial J_{a^{2}a^{2}}\partial J_{a^{3}a^{3}}}\right|_{J=0}$ are written below :
\begin{align}
F(\{1,2,3\},\emptyset,\emptyset)=&\left(-iV\frac{E_{a^{1}-1}}{\lambda}\right)\left(-iV\frac{E_{a^{2}-1}}{\lambda}\right)\left(-iV\frac{E_{a^{3}-1}}{\lambda}\right)A_{N}(z_{1},\ldots,z_{N})\frac{1}{\displaystyle\det_{1\leq p,q\leq N}\left(t_{p}^{q-1}\right)},
\end{align}
\begin{align}
F(\{i,l\},\{k\},\emptyset)=&\left(-iV\frac{E_{a^{i}-1}}{\lambda}\right)\left(-iV\frac{E_{a^{l}-1}}{\lambda}\right)\left(-\frac{V}{(\lambda V)^{\frac{1}{3}}}\right)\partial_{a^{k}}A_{N}(z_{1},\ldots,z_{N})\frac{1}{\displaystyle\det_{1\leq p,q\leq N}\left(t_{p}^{q-1}\right)},
\end{align}
\begin{align}
F(\{i\},\{l,k\},\emptyset)=&\left(-iV\frac{E_{a^{i}-1}}{\lambda}\right)\left(-\frac{V}{(\lambda V)^{\frac{1}{3}}}\right)^{2}\partial_{a^{l}}\partial_{a^{k}}A_{N}(z_{1},\ldots,z_{N})\frac{1}{\displaystyle\det_{1\leq p,q\leq N}\left(t_{p}^{q-1}\right)},
\end{align}
\begin{align}
F(\emptyset,\{1,2,3\},\emptyset)=&\left(-\frac{V}{(\lambda V)^{\frac{1}{3}}}\right)^{3}\partial_{a^{1}}\partial_{a^{2}}\partial_{a^{3}}A_{N}(z_{1},\ldots,z_{N})\frac{1}{\displaystyle\det_{1\leq p,q\leq N}\left(t_{p}^{q-1}\right)},
\end{align}
\begin{align}
&F(\{i\},\{l\},\{k\})\notag\\
=&\left(-iV\frac{E_{a^{i}-1}}{\lambda}\right)\left(-\frac{V}{(\lambda V)^{\frac{1}{3}}}\right)\partial_{a^{l}}A_{N}(z_{1},\ldots,z_{N})\frac{-1}{\displaystyle\det_{1\leq p,q\leq N}\left(t_{p}^{q-1}\right)}\sum_{r=1, r\neq a^{k}}^{N}\frac{1}{t_{a^{k}}-t_{r}},
\end{align}
\begin{align}
F(\{i,l\},\emptyset,\{k\})=&\left(-iV\frac{E_{a^{i}-1}}{\lambda}\right)\left(-iV\frac{E_{a^{l}-1}}{\lambda}\right)A_{N}(z_{1},\ldots,z_{N})\frac{-1}{\displaystyle\det_{1\leq p,q\leq N}\left(t_{p}^{q-1}\right)}\sum_{r=1, r\neq a^{k}}^{N}\frac{1}{t_{a^{k}}-t_{r}},
\end{align}
\begin{align}
F(\emptyset,\{i,l\},\{k\})=&\left(-\frac{V}{(\lambda V)^{\frac{1}{3}}}\right)^{2}\partial_{a^{i}}\partial_{a^{l}}A_{N}(z_{1},\ldots,z_{N})\frac{-1}{\displaystyle\det_{1\leq p,q\leq N}\left(t_{p}^{q-1}\right)}\sum_{r=1, r\neq a^{k}}^{N}\frac{1}{t_{a^{k}}-t_{r}},
\end{align}
\begin{align}
F(\{i\},\emptyset,\{l,k\})=&\left(-iV\frac{E_{a^{i}-1}}{\lambda}\right)A_{N}(z_{1},\ldots,z_{N})\notag\\
&\frac{1}{\displaystyle\det_{1\leq p,q\leq N}\left(t_{p}^{q-1}\right)}\left(\sum_{r=1,r\neq a^{l}}^{N}\frac{1}{t_{a^{l}}-t_{r}}\sum_{w=1,w\neq a^{k}}^{N}\frac{1}{t_{a^{k}}-t_{w}}-\frac{1}{(t_{a^{l}}-t_{a^{k}})^{2}}\right),
\end{align}
\begin{align}
F(\emptyset,\{i\},\{l,k\})=&\left(-\frac{V}{(\lambda V)^{\frac{1}{3}}}\right)\partial_{a^{i}}A_{N}(z_{1},\ldots,z_{N})\notag\\
&\frac{1}{\displaystyle\det_{1\leq p,q\leq N}\left(t_{p}^{q-1}\right)}\left(\sum_{r=1,r\neq a^{l}}^{N}\frac{1}{t_{a^{l}}-t_{r}}\sum_{w=1,w\neq a^{k}}^{N}\frac{1}{t_{a^{k}}-t_{w}}-\frac{1}{(t_{a^{l}}-t_{a^{k}})^{2}}\right),
\end{align}
\begin{align}
F(\emptyset,\emptyset,\{1,2,3\})=&-\frac{1}{\displaystyle\det_{1\leq p,q\leq N}\left(t_{p}^{q-1}\right)}\sum_{r,w,f=1,r\neq a^{3},w\neq a^{2},f\neq a^{1}}^{N}\frac{1}{t_{a^{3}}-t_{r}}\frac{1}{t_{a^{2}}-t_{w}}\frac{1}{t_{a^{1}}-t_{f}}A_{N}(z_{1},\ldots,z_{N})\notag\\
&+\frac{1}{\displaystyle\det_{1\leq p,q\leq N}\left(t_{p}^{q-1}\right)}\sum_{r=1,r\neq a^{1}}^{N}\frac{1}{(t_{a^{2}}-t_{a^{3}})^{2}}\frac{1}{t_{a^{1}}-t_{r}}A_{N}(z_{1},\ldots,z_{N})\notag\\
&+\frac{1}{\displaystyle\det_{1\leq p,q\leq N}\left(t_{p}^{q-1}\right)}\sum_{w=1,w\neq a^{2}}^{N}\frac{1}{(t_{a^{1}}-t_{a^{3}})^{2}}\frac{1}{t_{a^{2}}-t_{w}}A_{N}(z_{1},\ldots,z_{N})\notag\\
&+\frac{1}{\displaystyle\det_{1\leq p,q\leq N}\left(t_{p}^{q-1}\right)}\sum_{f=1,f\neq a^{3}}^{N}\frac{1}{(t_{a^{1}}-t_{a^{2}})^{2}}\frac{1}{t_{a^{3}}-t_{f}}A_{N}(z_{1},\ldots,z_{N}).\notag\\
\end{align}

\newpage

If we sum up all the cases for sets $S$, $M$, and $\overline{M}$ that we have calculated so far and multiply by $\displaystyle\frac{1}{\mathcal{Z}[0]}$, we get the result of (\ref{v}):
\begin{align}
&\frac{1}{\mathcal{Z}[0]}\frac{\partial^{3}\mathcal{Z}[J]}{\partial J_{a^{1}a^{1}}\partial J_{a^{2}a^{2}}\partial J_{a^{3}a^{3}}}\biggl|_{J=0}\notag\\
=&\left(\frac{\displaystyle\det_{1\leq p,q\leq N}\left(t_{p}^{q-1}\right)}{A_{N}(z_{1},\ldots,z_{N})}\right)\Biggl\{F(\{1,2,3\},\emptyset,\emptyset)+F(\emptyset,\{1,2,3\},\emptyset)+F(\emptyset,\emptyset,\{1,2,3\})\notag\\
&+\sum_{i,l,k=1,i\neq l\neq k\neq i}^{3}\Biggl(F(\{i\},\{l\},\{k\})+\frac{F(\{i,l\},\{k\},\emptyset)}{2}+\frac{F(\{i\},\{l,k\},\emptyset)}{2}+\frac{F(\{i,l\},\emptyset,\{k\})}{2}\notag\\
&+\frac{F(\emptyset,\{i,l\},\{k\})}{2}+\frac{F(\{i\},\emptyset,\{l,k\})}{2}+\frac{F(\emptyset,\{i\},\{l,k\})}{2}\Biggl)\Biggl\}.
\end{align}

\item We consider the case $\pi=\{\{i\},\{l,k\}\}$, $|\pi|=2$, and $B=\{i\},\{l,k\}$, then
\begin{align}
&\left.\left(\frac{d}{dx}\right)^{|\pi|}(\log x)\right|_{x=\mathcal{Z}[0]}\left.\prod_{B\in \pi}\frac{\partial^{|B|}\mathcal{Z}[J]}{\displaystyle\prod_{j\in B}\partial J_{a^{j}a^{j}}}\right|_{J=0}\notag\\
=&\left(-\frac{1}{\mathcal{Z}[0]^{2}}\right)\left(\left.\frac{\partial\mathcal{Z}[J]}{\partial J_{a^{i}a^{i}}}\right|_{J=0}\right)\left(\left.\frac{\partial^{2}\mathcal{Z}[J]}{\partial J_{a^{l}a^{l}}\partial J_{a^{k}a^{k}}}\right|_{J=0}\right).\label{uuuua}
\end{align}

The calculations required to calculate $\displaystyle\left(-\frac{1}{\mathcal{Z}[0]^{2}}\right)\left(\left.\frac{\partial\mathcal{Z}[J]}{\partial J_{a^{i}a^{i}}}\right|_{J=0}\right)\left(\left.\frac{\partial^{2}\mathcal{Z}[J]}{\partial J_{a^{l}a^{l}}\partial J_{a^{k}a^{k}}}\right|_{J=0}\right)$ are written below :
\begin{align}
F(\{i\},\emptyset,\emptyset)=&-iV\frac{E_{a^{i}-1}}{\lambda}A_{N}(z_{1},\ldots,z_{N})\frac{1}{\displaystyle\det_{1\leq p,q\leq N}\left(t_{p}^{q-1}\right)}.
\end{align}
\begin{align}
F(\emptyset,\{i\},\emptyset)=&\left(-\frac{V}{(\lambda V)^{\frac{1}{3}}}\right)\partial_{a^{i}}A_{N}(z_{1},\ldots,z_{N})\frac{1}{\displaystyle\det_{1\leq p,q\leq N}\left(t_{p}^{q-1}\right)}.
\end{align}
\begin{align}
F(\emptyset,\emptyset,\{i\})=&A_{N}(z_{1},\ldots,z_{N})\frac{-1}{\displaystyle\det_{1\leq p,q\leq N}\left(t_{p}^{q-1}\right)}\sum_{r=1,r\neq a^{i}}^{N}\frac{1}{t_{a^{i}}-t_{r}}.
\end{align}
\begin{align}
F(\{l,k\},\emptyset,\emptyset)=&\left(-iV\frac{E_{a^{l}-1}}{\lambda}\right)\left(-iV\frac{E_{a^{k}-1}}{\lambda}\right)A_{N}(z_{1},\ldots,z_{N})\frac{1}{\displaystyle\det_{1\leq p,q\leq N}\left(t_{p}^{q-1}\right)},
\end{align}
\begin{align}
F(\{l\},\{k\},\emptyset)=&\left(-iV\frac{E_{a^{l}-1}}{\lambda}\right)\left(-\frac{V}{(\lambda V)^{\frac{1}{3}}}\right)\partial_{a^{k}}A_{N}(z_{1},\ldots,z_{N})\frac{1}{\displaystyle\det_{1\leq p,q\leq N}\left(t_{p}^{q-1}\right)},
\end{align}
\begin{align}
F(\emptyset,\{l,k\},\emptyset)=&\left(-\frac{V}{(\lambda V)^{\frac{1}{3}}}\right)^{2}\partial_{a^{l}}\partial_{a^{k}}A_{N}(z_{1},\ldots,z_{N})\frac{1}{\displaystyle\det_{1\leq p,q\leq N}\left(t_{p}^{q-1}\right)},
\end{align}
\begin{align}
F(\{l\},\emptyset,\{k\})=&\left(-iV\frac{E_{a^{l}-1}}{\lambda}\right)A_{N}(z_{1},\ldots,z_{N})\frac{-1}{\displaystyle\det_{1\leq p,q\leq N}\left(t_{p}^{q-1}\right)}\sum_{r=1,r\neq a^{k}}^{N}\frac{1}{t_{a^{k}}-t_{r}},
\end{align}
\begin{align}
F(\emptyset,\{l\},\{k\})=&\left(-\frac{V}{(\lambda V)^{\frac{1}{3}}}\right)\partial_{a^{l}}A_{N}(z_{1},\ldots,z_{N})\frac{-1}{\displaystyle\det_{1\leq p,q\leq N}\left(t_{p}^{q-1}\right)}\sum_{r=1,r\neq a^{k}}^{N}\frac{1}{t_{a^{k}}-t_{r}},
\end{align}
\begin{align}
&F(\emptyset,\emptyset,\{l,k\})\notag\\
=&A_{N}(z_{1},\ldots,z_{N})\frac{1}{\displaystyle\det_{1\leq p,q\leq N}\left(t_{p}^{q-1}\right)}\left(\sum_{r=1,r\neq a^{l}}^{N}\frac{1}{t_{a^{l}}-t_{r}}\sum_{j=1,j\neq a^{k}}^{N}\frac{1}{t_{a^{k}}-t_{j}}-\frac{1}{(t_{a^{l}}-t_{a^{k}})^{2}}\right).
\end{align}

If we sum up all the cases for sets $S$, $M$, and $\overline{M}$ that we have calculated so far and multiply by $\displaystyle-\frac{1}{\mathcal{Z}[0]^{2}}$, we get the result of (\ref{uuuua}):
\begin{align}
&\left(-\frac{1}{\mathcal{Z}[0]^{2}}\right)\left(\left.\frac{\partial\mathcal{Z}[J]}{\partial J_{a^{i}a^{i}}}\right|_{J=0}\right)\left(\left.\frac{\partial^{2}\mathcal{Z}[J]}{\partial J_{a^{l}a^{l}}\partial J_{a^{k}a^{k}}}\right|_{J=0}\right)\notag\\
=&-\left(\frac{\displaystyle\det_{1\leq p,q\leq N}\left(t_{p}^{q-1}\right)}{A_{N}(z_{1},\ldots,z_{N})}\right)^{2}\sum_{i,l,k=1,i\neq l\neq k\neq i}^{3}\Biggl[\Biggl\{F(\{i\},\emptyset,\emptyset)+F(\emptyset,\{i\},\emptyset)+F(\emptyset,\emptyset,\{i\})\Biggl\}\notag\\
&\Biggl\{F(\{l\},\{k\},\emptyset)+F(\{l\},\emptyset,\{k\})+F(\emptyset,\{l\},\{k\})+\frac{F(\emptyset,\emptyset,\{l,k\})}{2}+\frac{F(\{l,k\},\emptyset,\emptyset)}{2}+\frac{F(\emptyset,\{l,k\},\emptyset)}{2}\Biggl\}\Biggl].\label{y}
\end{align}

\item We consider the case $\pi=\{\{1\},\{2\},\{3\}\}$, $|\pi|=3$, and $B=\{1\},\{2\},\{3\}$, then 
\begin{align}
&\left.\left(\frac{d}{dx}\right)^{|\pi|}(\log x)\right|_{x=\mathcal{Z}[0]}\left.\prod_{B\in \pi}\frac{\partial^{|B|}\mathcal{Z}[J]}{\displaystyle\prod_{j\in B}\partial J_{a^{j}a^{j}}}\right|_{J=0}\notag\\
=&\frac{2}{\mathcal{Z}[0]^{3}}\left(\left.\frac{\partial\mathcal{Z}[J]}{\partial J_{a^{1}a^{1}}}\right|_{J=0}\right)\left(\left.\frac{\partial\mathcal{Z}[J]}{\partial J_{a^{2}a^{2}}}\right|_{J=0}\right)\left(\left.\frac{\partial\mathcal{Z}[J]}{\partial J_{a^{3}a^{3}}}\right|_{J=0}\right).\label{uuuuu}
\end{align}

The calculations required to calculate $\displaystyle\frac{2}{\mathcal{Z}[0]^{3}}\left(\left.\frac{\partial\mathcal{Z}[J]}{\partial J_{a^{1}a^{1}}}\right|_{J=0}\right)\left(\left.\frac{\partial\mathcal{Z}[J]}{\partial J_{a^{2}a^{2}}}\right|_{J=0}\right)\left(\left.\frac{\partial\mathcal{Z}[J]}{\partial J_{a^{3}a^{3}}}\right|_{J=0}\right)$ are written below :
\begin{align}
F(\{i\},\emptyset,\emptyset)=&-iV\frac{E_{a^{i}-1}}{\lambda}A_{N}(z_{1},\ldots,z_{N})\frac{1}{\displaystyle\det_{1\leq p,q\leq N}\left(t_{p}^{q-1}\right)}.
\end{align}
\begin{align}
F(\emptyset,\{i\},\emptyset)=&\left(-\frac{V}{(\lambda V)^{\frac{1}{3}}}\right)\partial_{a^{i}}A_{N}(z_{1},\ldots,z_{N})\frac{1}{\displaystyle\det_{1\leq p,q\leq N}\left(t_{p}^{q-1}\right)}.
\end{align}
\begin{align}
F(\emptyset,\emptyset,\{i\})=&A_{N}(z_{1},\ldots,z_{N})\frac{-1}{\displaystyle\det_{1\leq p,q\leq N}\left(t_{p}^{q-1}\right)}\sum_{r=1,r\neq a^{i}}^{N}\frac{1}{t_{a^{i}}-t_{r}}.
\end{align}

If we sum up all the cases for sets $S$, $M$, and $\overline{M}$ that we have calculated so far and multiply by $\displaystyle\frac{2}{\mathcal{Z}[0]^{3}}$, we get the result of (\ref{uuuuu}):
\begin{align}
&\frac{2}{\mathcal{Z}[0]^{3}}\left(\left.\frac{\partial\mathcal{Z}[J]}{\partial J_{a^{1}a^{1}}}\right|_{J=0}\right)\left(\left.\frac{\partial\mathcal{Z}[J]}{\partial J_{a^{2}a^{2}}}\right|_{J=0}\right)\left(\left.\frac{\partial\mathcal{Z}[J]}{\partial J_{a^{3}a^{3}}}\right|_{J=0}\right)\notag\\
=&2\left(\frac{\displaystyle\det_{1\leq l,j\leq N}\left(t_{l}^{j-1}\right)}{A_{N}(z_{1},\ldots,z_{N})}\right)^{3}\prod_{i=1}^{3}\Biggl\{F(\{i\},\emptyset,\emptyset)+F(\emptyset,\{i\},\emptyset)+F(\emptyset,\emptyset,\{i\})\Biggl\}.\label{z}
\end{align}

 From i), ii) and iii), all results up to now are combined to obtain the calculation result of the three point functions $G_{|a^{1}|a^{2}|a^{3}|}$. 
\begin{align}
&G_{|a^{1}|a^{2}|a^{3}|}\notag\\
=&\left.\frac{iV}{\mathcal{Z}[0]}\frac{\partial^{3}\mathcal{Z}[J]}{\partial J_{a^{1}a^{1}}\partial J_{a^{2}a^{2}}\partial J_{a^{3}a^{3}}}\right|_{J=0}-\frac{iV}{\mathcal{Z}[0]^{2}}\left(\left.\frac{\partial\mathcal{Z}[J]}{\partial J_{a^{1}a^{1}}}\right|_{J=0}\right)\left(\left.\frac{\partial^{2}\mathcal{Z}[J]}{\partial J_{a^{2}a^{2}}\partial J_{a^{3}a^{3}}}\right|_{J=0}\right)\notag\\
&-\frac{iV}{\mathcal{Z}[0]^{2}}\left(\left.\frac{\partial\mathcal{Z}[J]}{\partial J_{a^{2}a^{2}}}\right|_{J=0}\right)\left(\left.\frac{\partial^{2}\mathcal{Z}[J]}{\partial J_{a^{1}a^{1}}\partial J_{a^{3}a^{3}}}\right|_{J=0}\right)-\frac{iV}{\mathcal{Z}[0]^{2}}\left(\left.\frac{\partial\mathcal{Z}[J]}{\partial J_{a^{3}a^{3}}}\right|_{J=0}\right)\left(\left.\frac{\partial^{2}\mathcal{Z}[J]}{\partial J_{a^{1}a^{1}}\partial J_{a^{2}a^{2}}}\right|_{J=0}\right)\notag\notag\\
&+\frac{2iV}{\mathcal{Z}[0]^{3}}\left(\left.\frac{\partial\mathcal{Z}[J]}{\partial J_{a^{1}a^{1}}}\right|_{J=0}\right)\left(\left.\frac{\partial\mathcal{Z}[J]}{\partial J_{a^{2}a^{2}}}\right|_{J=0}\right)\left(\left.\frac{\partial\mathcal{Z}[J]}{\partial J_{a^{3}a^{3}}}\right|_{J=0}\right)\notag\\
=&(iV)\left(\frac{\displaystyle\det_{1\leq p,q\leq N}\left(t_{p}^{q-1}\right)}{A_{N}(z_{1},\ldots,z_{N})}\right)\notag\\
&\Biggl\{F(\{1,2,3\},\emptyset,\emptyset)+F(\emptyset,\{1,2,3\},\emptyset)+F(\emptyset,\emptyset,\{1,2,3\})\notag\\
&+\sum_{i,l,k=1,i\neq l\neq k\neq i}^{3}\Biggl(F(\{i\},\{l\},\{k\})+\frac{F(\{i,l\},\{k\},\emptyset)}{2}+\frac{F(\{i\},\{l,k\},\emptyset)}{2}+\frac{F(\{i,l\},\emptyset,\{k\})}{2}\notag\\
&+\frac{F(\emptyset,\{i,l\},\{k\})}{2}+\frac{F(\{i\},\emptyset,\{l,k\})}{2}+\frac{F(\emptyset,\{i\},\{l,k\})}{2}\Biggl)\Biggl\}\notag\\
&-(iV)\left(\frac{\displaystyle\det_{1\leq p,q\leq N}\left(t_{p}^{q-1}\right)}{A_{N}(z_{1},\ldots,z_{N})}\right)^{2}\sum_{i,l,k=1,i\neq l\neq k\neq i}^{3}\Biggl[\Biggl\{F(\{i\},\emptyset,\emptyset)+F(\emptyset,\{i\},\emptyset)+F(\emptyset,\emptyset,\{i\})\Biggl\}\notag\\
&\Biggl\{F(\{l\},\{k\},\emptyset)+F(\{l\},\emptyset,\{k\})+F(\emptyset,\{l\},\{k\})+\frac{F(\emptyset,\emptyset,\{l,k\})}{2}+\frac{F(\{l,k\},\emptyset,\emptyset)}{2}+\frac{F(\emptyset,\{l,k\},\emptyset)}{2}\Biggl\}\Biggl]\notag\\
&+(2iV)\left(\frac{\displaystyle\det_{1\leq p,q\leq N}\left(t_{p}^{q-1}\right)}{A_{N}(z_{1},\ldots,z_{N})}\right)^{3}\left(\prod_{i=1}^{3}\Biggl\{F(\{i\},\emptyset,\emptyset)+F(\emptyset,\{i\},\emptyset)+F(\emptyset,\emptyset,\{i\})\Biggl\}\right).\label{z}
\end{align}

\end{enumerate}


\end{document}